\documentclass[draft,onecolumn]{IEEEtran}

\usepackage{amsmath}

\usepackage{cite}
\usepackage{amsmath,amssymb,amsfonts,amsthm}
\usepackage{txfonts}
\usepackage{algorithmic}
\usepackage{graphicx}
\usepackage{caption}
\usepackage{textcomp}
\usepackage{xcolor}
\usepackage{tikz}
\usetikzlibrary{positioning, shapes.geometric}
\usepackage{setspace} 
\usepackage{comment}
\usepackage{url}
\usepackage{blkarray}

\newtheorem{theorem}{Theorem}
\newtheorem*{theorem*}{Theorem}
\newtheorem{definition}{Definition}
\newtheorem{corollary}{Corollary}
\newtheorem{lemma}{Lemma}
\newtheorem*{lemma*}{Lemma}
\newtheorem{example}{Example}
\newtheorem{remark}{Remark}

\begin{document}
\title{Function-Correcting Codes With Data Protection}

\author{\IEEEauthorblockN{ Charul Rajput\IEEEauthorrefmark{1}, B. Sundar Rajan\IEEEauthorrefmark{2}, Ragnar Freij-Hollanti\IEEEauthorrefmark{1}, Camilla Hollanti\IEEEauthorrefmark{1}}

\IEEEauthorblockA{\IEEEauthorrefmark{1}Department of Mathematics and Systems Analysis, Aalto University, Finland
    \\\{charul.rajput, ragnar.freij, camilla.hollanti\}@aalto.fi}\\
    \IEEEauthorblockA{\IEEEauthorrefmark{2}Department of Electrical Communication Engineering, Indian Institute of Science, \\ Bengaluru, India
    \\\ bsrajan@iisc.ac.in}}

\maketitle

\begin{abstract}
Function-correcting codes (FCCs) are designed to provide error protection for the value of a function computed on the data. Existing work typically focuses solely on protecting the function value and not the underlying data. In this work, we propose a general framework that offers protection for both the data and the function values. Since protecting the data inherently contributes to protecting the function value, we focus on scenarios where the function value requires stronger protection than the data itself. A two-step construction procedure for such codes is proposed, and bounds on the optimal redundancy of general FCCs with data protection are reported. Using these results, we exhibit examples that show that data protection can be added to existing FCCs without increasing redundancy. Using our two-step construction procedure, we present explicit constructions of FCCs with data protection for specific families of functions, such as locally bounded functions and the  Hamming weight function. We associate a graph called {\it minimum-distance graph} to a code and use it to show that perfect codes and maximum distance separable (MDS) codes cannot provide additional protection to function values over and above the amount of protection for data for any function. 
Then we focus on linear FCCs and provide some results for linear functions, leveraging their inherent structural properties. While FCCs for linear functions have been considered earlier in the literature, to the best of our knowledge, the linearity of the FCC itself has not been studied before.
Finally, we generalize the Plotkin and Hamming bounds well known in classical error-correcting coding theory to FCCs with data protection. 
\end{abstract}


% \fi
\IEEEpeerreviewmaketitle

\section{Introduction}
\label{intro}

Function-correcting codes (FCCs), introduced by \cite{LBWY2023}, are designed to safeguard specific functions or attributes of a message vector rather than the entire vector itself. These codes operate in scenarios where data is transmitted from a sender to a receiver over a potentially error-prone channel. Unlike traditional error-correcting codes, which aim to protect the full message against a certain number of errors, FCCs focus solely on preserving the integrity of a particular function derived from the data. This targeted protection is especially valuable in contexts where only specific aspects of the data are of interest to the receiver, such as in certain machine learning applications or data storage systems. By concentrating on the function rather than the whole message, FCCs offer a more efficient approach to error correction.

Error correction works by introducing redundancy into the data, enabling the receiver to detect and correct a certain number of errors during decoding. Generally, the level of redundancy required increases with the number of errors that need to be corrected. However, when the function of interest has a smaller co-domain relative to its domain, it becomes possible to preserve its values with less redundancy, even in the presence of the same number of errors. This efficiency highlights the significance of a function-correcting code (FCC), which is designed to protect specific function rather than the entire dataset, offering a more efficient solution in such scenarios.

\subsection{Related works}
Several related works \cite{BC2011, MW1967,BK1981,AC1981,OR2001,BNZ2009,KW2017,SSGAGD2019} explore the idea of prioritizing certain parts of data, such as unequal error protection codes. The formal foundation of function-correcting codes is established in \cite{LBWY2023}, where the authors introduce key parameters and structural definitions for these codes. They also derive upper and lower bounds on the redundancy required for FCCs. Furthermore, they construct FCCs specifically for certain classes of functions, including locally binary functions, the Hamming weight function, and Hamming weight distribution functions. Some of these constructions prove to be optimal, achieving the theoretical lower bounds on redundancy. This initial work focuses on communication over binary symmetric channels. Subsequently, the concept of FCCs is extended in \cite{XLC2024} to symbol-pair read channels over binary fields, and is later generalized to symbol-pair read channels over finite fields in \cite{SSY2025}.

In \cite{PR2025}, the authors focus on linear functions, which provide a more structured view of the function's domain, as the kernel of a linear function forms a subspace of the domain $\mathbb{F}_q^k$. Using the coset partition of this kernel, they present a construction of FCCs for linear functions. The paper also derives a lower bound on the redundancy required for FCCs, which leads to a corresponding bound in classical error-correcting codes, since for a bijective function, an FCC is equivalent to an error-correcting code (ECC). They demonstrate the tightness of this redundancy bound for certain parameters.
The work \cite{ZXZG2025} by Zhang et al. specifically considers two types of functions, the Hamming weight function and the Hamming weight distribution function, offering improved bounds and constructions of FCCs for these functions.

In a recent work \cite{RRHH2025}, the authors generalize the definition of a locally binary function to a locally $(\rho, \lambda)$-bounded function, and then provide an upper bound on the redundancy of FCCs for these functions. They also prove that any function on $\mathbb{F}_2^k$ can be considered a locally $(\rho, \lambda)$-bounded function for carefully chosen parameters, making this framework applicable to arbitrary functions. Furthermore, an extension of this work to the $b$-symbol read channel is presented in \cite{VSS2026}. Another study \cite{LS2025} provides an upper bound on the redundancy of FCCs that is within a logarithmic factor of a known lower bound given in \cite{LBWY2023}. The authors also prove that this lower bound is tight for sufficiently large field sizes.

\subsection{Motivation}
By the definition of FCCs given in all existing works, if two message vectors map to the same function value, the receiver does not need to distinguish between their corresponding codewords. This implies that the distance between such codewords is irrelevant. However, if two codewords correspond to different function values, a specific minimum distance between them is required to ensure the protection of the function. As a result, an FCC does not guarantee any form of error protection for the data itself.

All existing work in the literature focuses solely on protecting the function values, not the data. In this work, we present a general framework in which both the data and the function values are assigned distinct levels of error protection. For instance, suppose that we require protection for the data against up to $t_d$  errors and for the function values against up to $t_f$  errors. Since protecting the data inherently contributes to protecting the function, it is only meaningful to consider the case where the function requires stronger protection than the data, i.e., $t_d<t_f$. The general setup for these codes is shown in Fig. \ref{fig:1}.

\begin{figure*}[ht]
    \centering
    \resizebox{\textwidth}{!}{%
    \begin{tikzpicture}[>=latex, thick, node distance=2cm and 1.5cm]

% Nodes
\node (a) { };

% Encoder block
\node[draw, minimum width=2cm, minimum height=1.5cm, right=of a] (encoder) {Encoder};
\node[below=0.2cm of encoder] {\small Transmitter};

% Channel block
\node[draw, fill=blue!20, minimum width=2.8cm, minimum height=1.2cm, right=1.8cm of encoder] (channel) {Channel};

% Receiver outer box
\node[draw, minimum width=5.5cm, minimum height=3.5cm, right=2cm of channel] (receiver) {};

% Receiver label
\node[below=0.2cm of receiver] {\small Receiver};

% Decoder blocks inside receiver
\node[draw, minimum width=4.5cm, minimum height=1.2cm] at ([yshift=0.8cm]receiver.center) (funcdec) {Function value decoder};
\node[draw, minimum width=4.5cm, minimum height=1.2cm] at ([yshift=-0.8cm]receiver.center) (datadec) {Data decoder};

% Output labels
\node[right=1.2cm of funcdec] (fy) {$\hat{f({u})}$};
\node[right=1.2cm of datadec] (ahat) {$\hat{u}$};

% Entry point into receiver (split point)
\coordinate (entry) at ([xshift=-2.75cm]receiver.center);

% Arrows
\draw[->] (a) -- node[above] {${u}$}  (encoder);
\draw[->] (encoder) -- node[above] {$c({u})$} (channel);
\draw[-] (channel) -- node[above] {$c({u})+{e}$} (entry);
%\draw[->] (noise) -- (entry);

% Fan-out inside receiver
\draw[->] (entry) -- (funcdec.west);
\draw[->] (entry) -- (datadec.west);

% Outputs
\draw[->] (funcdec) -- (fy);
\draw[->] (datadec) -- (ahat);

\end{tikzpicture}}
    \caption{General setup of function-correcting codes. The transmitter has a message vector ${u}$, and the receiver has a special interest in a function $f({u})$ of this message. The transmitter encodes the message ${u}$ to $c({u})$ and sends it over the erroneous channel. The receiver receives the vector $c({u}) + {e}$, and can decode the message and the function value with different levels of error protection guarantees based on the importance of the function $f$ for the receiver.}
    \label{fig:1}
\end{figure*}
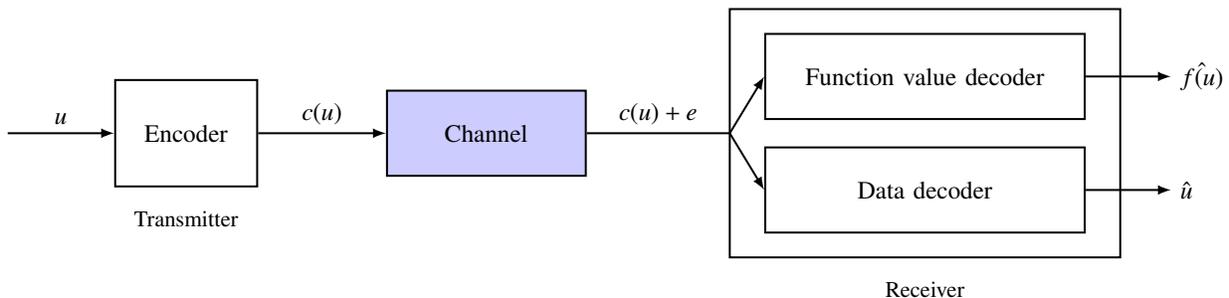

There is another way to look at it. When we construct a code based on the original definition of FCCs, where the sole focus is on protecting the function, multiple FCCs may exist with the same amount of redundancy. These codes will perform equally well in terms of correcting the function value, but they may offer varying levels of data protection or none at all. This difference in data protection across different codes cannot be captured unless we also consider the overall minimum distance. This distinction is illustrated more clearly in Example \ref{ex2}. Therefore, the general framework of FCCs introduced in this paper is essential to capture this property.

\subsubsection{In communication networks}
In today's world, many communications occur within large, interconnected networks. If the communication is one-to-one and the receiver is only interested in a function of the data, not the data itself, it may be more efficient to encode and transmit only the function values. However, when nodes are connected in large networks, such as the wireless and wired networks illustrated in Fig. \ref{fig:2}, they often need to forward data to other connected nodes.

A key limitation in the original function-correcting code setup is that, while data is transmitted from the sender to the receiver, error protection for the data was not explicitly addressed. To enable reliable data flow in a network, some level of error protection for the data itself is necessary. Even if intermediate nodes do not require the full data, different nodes in the network may be interested in different functions or attributes of the data. This possible diversity of interests makes the reliable flow of data throughout the network essential.

\begin{figure}[ht]
    \centering
    \begin{minipage}{0.33\textwidth}
\centering
\resizebox{\linewidth}{!}{
    \begin{tikzpicture}[
  server/.style={rectangle, draw=black, fill=white, minimum width=2cm, minimum height=1cm},
  device/.style={circle, draw=black, fill=white, minimum size=1.2cm, font=\small},
  line/.style={->, thick}
]

% Server
\node[server] (server) {Server};

%\coordinate (join) at ([yshift=-1.25cm]server);

% First row of devices: A1 to A4
\node[device, below left=2.5cm and 2.5cm of server] (A1) {$A_1$};
\node[device, right=1.5cm of A1] (A2) {$A_2$};
\node[device, right=1.5cm of A2] (A3) {$A_3$};
\node[device, right=1.5cm of A3] (A4) {$A_4$};

% Second row of devices: A5 to A10
\node[device, below left=2cm and 1.75cm of A1] (A5) {$A_5$};
\node[device, right=1.5cm of A5] (A6) {$A_6$};
\node[device, right=1.5cm of A6] (A7) {$A_7$};
\node[device, right=1.5cm of A7] (A8) {$A_8$};
\node[device, right=1.5cm of A8] (A9) {$A_9$};
\node[device, right=1.5cm of A9] (A10) {$A_{10}$};

% Connections from Server to A1-A4
%\draw[-] (server) -- (join);
\draw[line] (server) -- (A1);
\draw[line] (server) -- (A2);
\draw[line] (server) -- (A3);
\draw[line] (server) -- (A4);

% Connections from A1-A4 to A5-A10 (as per image structure)
\draw[line] (A1) -- (A5);
\draw[line] (A2) -- (A5);
\draw[line] (A1) -- (A6);
\draw[line] (A3) -- (A6);
\draw[line] (A1) -- (A7);
\draw[line] (A4) -- (A7);
\draw[line] (A2) -- (A8);
\draw[line] (A3) -- (A8);
\draw[line] (A2) -- (A9);
\draw[line] (A4) -- (A9);
\draw[line] (A3) -- (A10);
\draw[line] (A4) -- (A10);

\end{tikzpicture}}
\caption*{(a) Wired network}
\end{minipage}
\hfill
    \begin{minipage}{0.29\textwidth}
\centering
\resizebox{\linewidth}{!}{
\begin{tikzpicture}[
  server/.style={rectangle, draw=black, fill=white, minimum width=2cm, minimum height=1cm},
  device/.style={circle, draw=black, fill=white, minimum size=1.2cm, font=\small},
  line/.style={->, thick}
]
% Server
\node[server] (server) {Server};

\coordinate (join) at ([yshift=-1.25cm]server);

% First row of devices: A1 to A4
\node[device, below left=2.5cm and 2.5cm of server] (A1) {$A_1$};
\node[device, right=1.5cm of A1] (A2) {$A_2$};
\node[device, right=1.5cm of A2] (A3) {$A_3$};
\node[device, right=1.5cm of A3] (A4) {$A_4$};

% Second row of devices: A5 to A10
\node[device, below left=2.0cm and 0.9cm of A1] (A5) {$A_5$};
\node[device, right=.5cm of A5] (A6) {$A_6$};
\node[device, right=.5cm of A6] (A7) {$A_7$};
\node[device, below left=2cm and 0.9cm of A4] (A8) {$A_8$};
\node[device, right=0.5cm of A8] (A9) {$A_9$};
\node[device, right=0.5cm of A9] (A10) {$A_{10}$};

\coordinate (join1) at ([yshift=-1cm]A1);
\coordinate (join2) at ([yshift=-1cm]A4);

% Connections from Server to A1-A4
\draw[-] (server) -- (join);
\draw[line] (join) -- (A1);
\draw[line] (join) -- (A2);
\draw[line] (join) -- (A3);
\draw[line] (join) -- (A4);

% Connections from A1-A4 to A5-A10 (as per image structure)
\draw[-] (A1) -- (join1);
\draw[line] (join1) -- (A5);
\draw[line] (join1) -- (A6);
\draw[line] (join1) -- (A7);

\node[below right=2cm and 0.35cm of A2] (dots) {\large $\bullet\ \bullet\ \bullet$};

\draw[-] (A4) -- (join2);
\draw[line] (join2) -- (A8);
\draw[line] (join2) -- (A9);
\draw[line] (join2) -- (A10);

\end{tikzpicture}}
\caption*{(b) Wireless network}
\end{minipage}
\hfill
\begin{minipage}{0.33\textwidth}
\centering
\resizebox{\linewidth}{!}{
    \begin{tikzpicture}[
  server/.style={rectangle, draw=black, fill=white, minimum width=2cm, minimum height=1cm},
  device/.style={circle, draw=black, fill=white, minimum size=1.2cm, font=\small},
  line/.style={->, thick}
]

% Server
\node[server] (server) {Server};

\coordinate (join) at ([yshift=-1.5cm]server);

% First row of devices: A1 to A4
\node[device, below left=2.5cm and 2.5cm of server] (A1) {$A_1$};
\node[device, right=1.5cm of A1] (A2) {$A_2$};
\node[device, right=1.5cm of A2] (A3) {$A_3$};
\node[device, right=1.5cm of A3] (A4) {$A_4$};

% Second row of devices: A5 to A10
\node[device, below left=2cm and 1.75cm of A1] (A5) {$A_5$};
\node[device, right=1.5cm of A5] (A6) {$A_6$};
\node[device, right=1.5cm of A6] (A7) {$A_7$};
\node[device, right=1.5cm of A7] (A8) {$A_8$};
\node[device, right=1.5cm of A8] (A9) {$A_9$};
\node[device, right=1.5cm of A9] (A10) {$A_{10}$};

% Connections from Server to A1-A4
\draw[-] (server) -- (join);
\draw[line] (join) -- (A1);
\draw[line] (join) -- (A2);
\draw[line] (join) -- (A3);
\draw[line] (join) -- (A4);

% Connections from A1-A4 to A5-A10 (as per image structure)
\draw[line] (A1) -- (A5);
\draw[line] (A2) -- (A5);
\draw[line] (A1) -- (A6);
\draw[line] (A3) -- (A6);
\draw[line] (A1) -- (A7);
\draw[line] (A4) -- (A7);
\draw[line] (A2) -- (A8);
\draw[line] (A3) -- (A8);
\draw[line] (A2) -- (A9);
\draw[line] (A4) -- (A9);
\draw[line] (A3) -- (A10);
\draw[line] (A4) -- (A10);

\end{tikzpicture}}
\caption*{(c) Wireless and wired network}
\end{minipage}
\caption{General network setups}
    \label{fig:2}
\end{figure}

\subsubsection{In data storage}
The general setup of function-correcting codes defined in this paper is also beneficial for storage systems where certain attributes of the data are more important than others. While the entire dataset can be stored with basic error protection, the important attributes or features of the data can be stored with a higher level of protection. This naturally aligns with the idea behind Unequal Error Protection (UEP) \cite{BNZ2009}, where more critical parts of the data are prioritized in terms of reliability. In fact, UEP can be seen as a specific instance of FCCs, where the function of interest corresponds to selected bits or segments. FCCs extend this idea by allowing protection of more general functions, offering greater flexibility in how reliability is distributed across the data.

\subsection{Contributions}
The contributions of this paper are summarized as follows:
\begin{enumerate}
\item We first introduce a more general approach and framework of function-correcting codes that incorporates data protection along with protection for function values.  Then we propose a two-step construction procedure for such codes.
\item We provide bounds on the optimal redundancy of FCCs with data protection. Using these results we exhibit examples that show that without increasing the redundancy data protection can be added to the FCCs.
\item We associate a graph called {\it minimum-distance graph} to a code and using this it is shown that perfect codes and MDS codes cannot provide additional protection to function values over and above the amount of protection for data for any function. 
\item Using our two-step construction procedure, we present explicit constructions of FCCs with data protection for specific families of functions, such as locally bounded functions and Hamming weight function.
\item We then focus on \emph{linear FCCs} and derive some results for the case of linear functions by exploiting their structure. To the best of our knowledge, although FCCs for linear functions have appeared earlier in the literature, the case where the FCC itself is linear has not been studied before.
\item We also extend the classical Plotkin and Hamming bounds from error-correcting codes to the framework of FCCs with data protection.
\end{enumerate}

We note that parts of this work are also contained in two papers submitted to the 2026 IEEE International Symposium on Information Theory (ISIT). The first of these contains the results in Contributions~1 and~2 in abbreviated form. The second submitted paper builds on the minimum-distance graph introduced in Contribution~3 by considering an $\alpha$-distance graph as a further generalization, and includes some additional results in that direction. In contrast, the present journal manuscript gives the complete framework of FCCs with data protection, full proofs of all results, and the broader collection of contributions listed above, including Contributions~3--6.

\subsection{Paper organization}

Preliminaries constituting basic definitions and known results used in the paper subsequently are presented in Section \ref{preliminaries}. In Section \ref{main} a method of constructing FCCs that offer protection to the data (information symbols) along with offering protection for function values is presented. Bounds on the optimal redundancy for FCCs with data protection are obtained in  Section \ref{general_definitions}. 
Calling FCCs that offer strictly larger protection for function values than for the data {\it strict FCCs}, in Section \ref{sec5} it is shown that perfect codes and MDS codes cannot be {\it strict FCCs}. This is achieved by introducing a notion of {\it minimum-distance graph} for any code. In Section \ref{specific_functions}, we study FCCs with data protection  for three classes of  specific functions and apply our construction method to obtain FCCs designed for these functions. Section \ref{linearFCC} focuses on linear FCCs, with particular attention to linear functions due to their well known additional structure. In Section \ref{bounds}, we present generalizations of Plotkin and Hamming bounds on the redundancy of FCCs with data protection. Some results needed to obtain these bounds are relegated to Appendix. Finally,  Section \ref{conclusion} concludes the paper.

\subsection{Notation}
The finite field of size $q$ is denoted by $\mathbb{F}_q$, and $\mathbb{N}$ denotes the set of natural numbers. The space $\mathbb{F}_q^k$ denotes the $k$-dimensional vector space over $\mathbb{F}_q$.
We use the notation $\binom{n}{r} = \frac{n!}{r!(n - r)!}$ to denote the binomial coefficient. The set $\mathbb{N}^{m \times m}$ denotes the set of all $m \times m$ matrices with entries in the set of natural numbers $\mathbb{N}$. For any function $f : S \rightarrow T$, and for any $\alpha \in \mathrm{Im}(f)$, the preimage of $\alpha$ under $f$ is denoted by $f^{-1}(\alpha) = \{ s \in S \mid f(s) = \alpha \}$.  
For any vector $u \in \mathbb{F}_q^k$, $\mathrm{wt}(u)$ represents the number of non-zero entries in the vector $u$, and $d(u,v)$ denotes the Hamming distance between vectors $u$ and $v$, where $u,v \in \mathbb{F}_q^k$.  
 For any subset $J\subseteq [n]$, we write $\mathbb F_q^{J}$ for the set of all $J$-indexed vectors over $\mathbb F_q$. 
For any $u \in \mathbb{F}_q^n$ and radius $t \in \mathbb{N}$, the Hamming ball of radius $t$ centered at $u$ is defined as
$B(u,t) = \{ x \in \mathbb{F}_q^n \mid d(x,u) \le t \}.$
For any two sets $A$ and $B$, $A \setminus B$ denotes the set of elements that are in $A$ but not in $B$.
In examples throughout this paper, vectors in $\mathbb{F}_q^k$ are often written without commas, such as $0110 \in \mathbb{F}_2^4$.

\section{Preliminaries}
\label{preliminaries}

This section presents the basic concepts, definitions, and key results on function-correcting codes, based on the work in \cite{LBWY2023}. These foundations are important for the methods and discussions that follow in the rest of the paper.

\begin{definition}[Function-Correcting Codes]
Consider a function $f: \mathbb{F}_q^k \rightarrow \mathrm{Im}(f)$. A systematic encoding $\mathcal{C}: \mathbb{F}_q^k \rightarrow \mathbb{F}_q^{k+r}$ is defined as an $(f, t)$-function correcting code (FCC) if, for any $u_1, u_2 \in \mathbb{F}_q^k$ such that $f(u_1) \neq f(u_2)$, the following condition holds: 
$$d(\mathcal{C}(u_1), \mathcal{C}(u_2)) \geq 2t+1,$$
where $d(x, y)$ denotes the Hamming distance between vectors $x$ and $y$.
\end{definition}

%\begin{definition}[Optimal redundancy]
The \emph{optimal redundancy} $r_f(k, t)$ is defined as the minimum of $r$ for which there exists an $(f, t)$-FCC with an encoding function $\mathcal{C}: \mathbb{F}_q^k \rightarrow \mathbb{F}_q^{k+r}$.
%\end{definition}

\begin{definition}[Distance Requirement Matrix (DRM)]
Let $u_1, u_2,$ $\ldots, u_{M} \in \mathbb{F}_q^k$. The distance requirement matrix (DRM) $\mathcal{D}_f(t, u_1, u_2,\ldots, u_M)$ for an $(f, t)$-FCC is an $M \times M$ matrix with entries
$$
[\mathcal{D}_f(t, u_1, \ldots, u_M)]_{i, j} \\
= 
\begin{cases}
\max(2t+1-d(u_i, u_j), 0), & \text{if } f(u_i) \neq f(u_j), \\
0, & \text{otherwise},
\end{cases}
$$
where $i, j \in \{1,2, \ldots, M\}$.
\end{definition}

When $M=q^k,$ we get the DRM for $\mathcal{C}:\mathbb{F}_q^k \rightarrow \mathbb{F}_q^{k+r}.$

\begin{example}\label{ex1}
\normalfont Consider $\mathbb{F}_2^2 = \{00, 01, 10, 11\}$ and a function $f: \mathbb{F}_2^2 \to \{0, 1, 2\}$ such that $f(00)= 0, f(01)= f(10)=1, f(11)=2.$
Then for $t=1$ and $u_1=00,u_2=01,u_3=10,u_4=11$, we have the following DRMs:
$$
\mathcal{D}_f(t, u_1, u_2, u_3, u_4) = \begin{bmatrix} 0 & 2 & 2 & 1 \\
2 & 0& 0& 2 \\
2 &0 & 0 & 2 \\
1 & 2& 2& 0
\end{bmatrix};$$
$$\mathcal{D}_f(t, u_1, u_2, u_4) = \begin{bmatrix} 0 & 2 &  1 \\
	2 &0 &  2 \\
	1 & 2&  0
\end{bmatrix}; ~\mathcal{D}_f(t, u_1, u_2, u_3) = \begin{bmatrix} 0 & 2 & 2  \\
	2 & 0& 0 \\
	2 &0 & 0  
\end{bmatrix}. 
$$
\end{example}

\begin{definition}[Irregular-distance code or $\mathcal{D}$-code]
Let $\mathcal{D} \in \mathbb{N}^{M\times M}$. Then $\mathcal{P}=\{p_1, p_2, \ldots, p_M\}$ where $p_i \in \mathbb{F}_q^r, i=1,2,\ldots, M$  is said to be an irregular-distance code or $\mathcal{D}$-code if there is an ordering of $P$ such that $d(p_i, p_j) \geq [\mathcal{D}]_{i,j}$ for all $i, j \in \{1, 2, \ldots, M\}$. Further, $N(\mathcal{D})$ is defined as the smallest integer $r$ such that there exists a $\mathcal{D}$-code of length $r$. If $[\mathcal{D}]_{i, j} = D$ for all $i, j \in \{1,2,\ldots, M\}, i\neq j$, then $N(\mathcal{D})$ is denoted as $N(M, D)$.
\end{definition}

Here, $N(M,D)$ denotes the minimum length of an error-correcting code of size $M$ and minimum distance at least $D$. 
For $\mathcal{D}=\mathcal{D}_f(t, u_1, u_2,\ldots, u_{q^k})$, if we have a $\mathcal{D}$-code $\mathcal{P}=\{p_1, p_2, \ldots, p_{q^k}\}$ then we can use it to construct a $(f, t)$-FCC with the encoding $\mathcal{C}(u_i) = (u_i, p_i)$ for all $i \in \{1, 2, \ldots, q^k\}$.

\begin{example}
\normalfont Consider the same function  $f: \mathbb{F}_2^2 \to \{0, 1, 2\}$ from Example \ref{ex1}. Then we have   a $\mathcal{D}$-code $\mathcal{P}=\{000, 110, 110, 101\}$ for which distance between any pair of codewords is given by  
$$
 \begin{blockarray}{ccccc}
   &  000 & 110 & 110 & 101 \\
  \begin{block}{c[cccc]}
000 & 0 & 2 & 2 & 2 \\
110 & 2 & 0& 0& 2 \\
110 & 2 &0 & 0 & 2 \\
101 & 2 & 2& 2& 0 \\
    \end{block}
\end{blockarray}.$$
Since $r=3$ is the smallest length possible for a $\mathcal{D}$-code, we have $N(\mathcal{D}_f(t, u_1, u_2, u_3, u_4)) = 3$.
Further, the $(f, 1)$-FCC obtained using $\mathcal{P}$ is $\{00000, 01110, 10110, 11101\}.$
\end{example}

\begin{lemma}[{\cite[Lemma 2]{XLC2024}}]\label{lem:4from1}
    For any $M, D \in \mathbb{N}$ with $D \geq 10$ and $M \leq D^2$,
$$
N(M, D) \leq \frac{2D-2}{1 - 2\sqrt{\frac{\ln(D)}{D}}}.
$$
\end{lemma}

\begin{definition}
For a function $f: \mathbb{F}_q^k \to \mathrm{Im}(f)$, the distance between $f_i, f_j \in \mathrm{Im}(f)$ is defined as
$$d(f_i, f_j) = \min_{u_1, u_2 \in \mathbb{F}_q^k} \{d(u_1, u_2) |  f(u_1)=f_i, f(u_2) = f_j\}.$$
\end{definition}

\begin{definition}[Function Distance Matrix (FDM)]
Consider a function $f: \mathbb{F}_q^k \to \mathrm{Im}(f)$ and $E=|\mathrm{Im}(f)|$. Then the  $E \times E$ matrix $\mathcal{D}_f(t, f_1, f_2,\ldots, f_E)$ with entries given as
$$[\mathcal{D}_f(t, f_1, f_2,\ldots, f_E)]_{i, j} = \begin{cases}  \max(2t+1 -d(f_i, f_j), 0), \hfill\ \text{if} \ i \neq j, \\
0 \hfill \text{otherwise},
\end{cases}$$
is called the function distance matrix.
\end{definition}

\begin{example}
\normalfont For the function  $f: \mathbb{F}_2^2 \to \{0, 1,2\}$ given in Example \ref{ex1}, we have   
$$\mathcal{D}_f(t=1, f_1=0, f_2=1, f_3=2) = \begin{bmatrix} 0 & 2 & 1 \\
2 & 0 & 2 \\
1 & 2 & 0
\end{bmatrix}.$$
\end{example}

%\ch{Paper citations for quoted results usually done as below I think, should we change all of them? Should probably also add the Thm number from the paper in question like: \cite[Thm. X]{LBWY2023}.}\cra{done}
\begin{corollary}[{\cite[Corollary 1]{LBWY2023}}]\label{thm1}
For any function $f: \mathbb{F}_q^k \to \mathrm{Im}(f)$ and $\{u_1, u_2, \ldots, u_m\}\subseteq \mathbb{F}_q^k$,
$$r_f(k, t) \geq N(\mathcal{D}_f(t, u_1, u_2, \ldots, u_m)),$$
and for $|\mathrm{Im}(f)|\geq 2$, $r_f (k, t) \geq 2t$.
\end{corollary}

\begin{theorem}[{\cite[Theorem 2]{LBWY2023}}]\label{thm2}
For any function $f: \mathbb{F}_q^k \to \mathrm{Im}(f)=\{f_1, f_2, \ldots, f_E\}$,
$$r_f(k, t) \leq N(\mathcal{D}_f(t, f_1, f_2, \ldots, f_E)),$$
where $\mathcal{D}_f(t, f_1, f_2, \ldots, f_E)$ is a FDM.
\end{theorem}

\begin{corollary}[{\cite[Corollary 2]{LBWY2023}}]\label{col1}
If there exists a set of representative information vectors $u_1, u_2, \ldots, u_{E}$ with $\{f(u_1), f(u_2),$ $\ldots,$ $f(u_E)\} = \mathrm{Im}(f)$ and $\mathcal{D}_f(t, u_1, u_2, \ldots, u_E)=\mathcal{D}_f(t, f_1, f_2, \ldots, f_E)$, then
$$r_f(k, t) = N(\mathcal{D}_f(t, f_1, f_2, \ldots, f_E)).$$
\end{corollary}

Let $wt: \mathbb{F}_2^k \rightarrow \{0, 1, 2, \ldots, k\}$ denote the Hamming weight function, where $wt(u)$ represents the number of non-zero entries in the vector $u \in \mathbb{F}_2^k$. The following lower bound on the redundancy of a $(wt, t)$-FCC is given in \cite{LBWY2023}.

\begin{corollary}[{\cite[Corollary 3]{LBWY2023}}]\label{colHb}
    For any $k > t$,
\[
r_{wt}(k, t) \geq \frac{10t^3 + 30t^2 + 20t + 12}{3t^2 + 12t + 12}.
\]
\end{corollary}

\section{A Construction Procedure for FCCs with Data Protection}\label{main}

In this section, we introduce a general definition for function-correcting codes that simultaneously protect data and function values. We also describe a construction technique developed for this extended setting.

\begin{definition}
Consider a function $f: \mathbb{F}_q^k \rightarrow \mathrm{Im}(f)$. An encoding $\mathfrak{C}_f: \mathbb{F}_q^k \rightarrow \mathbb{F}_q^{k+r}$ is defined as an $(f\!:d_d,d_f)$-function correcting code (FCC) if, 
\begin{itemize}
    \item for any $u_1, u_2 \in \mathbb{F}_q^k$, such that $u_1 \ne u_2$,
$$d(\mathfrak{C}_f(u_1), \mathfrak{C}_f(u_2)) \geq d_d,$$
\item 
for any $u_1, u_2 \in \mathbb{F}_q^k$ such that $f(u_1) \neq f(u_2)$, the following condition holds: 
$$d(\mathfrak{C}_f(u_1), \mathfrak{C}_f(u_2)) \geq d_f,$$
\end{itemize}
where $d_d$ and $d_f$ are two non-negative integers such that $d_d \le d_f$, and $d(x, y)$ denotes the Hamming distance between vectors $x$ and $y$.
\end{definition}

For the $(f\!:d_d,d_f)$-FCC defined above, we will also use the alternative notation $(f, t_d, t_f)$-FCC in this paper, where $t_d = \left\lfloor \frac{d_d - 1}{2} \right\rfloor$ denotes the error-correction capability for data, and $t_f = \left\lfloor \frac{d_f - 1}{2} \right\rfloor$ represents the error-correction capability for the function values. We will use $t_d$ and $t_f$ interchangeably with $d_d$ and $d_f$, respectively, throughout the paper, with the understanding that they are related as described above. The notation `$(f, t_d, t_f)$-FCC' is inspired by the earlier terminology `$(f, t)$-FCC' used in \cite{LBWY2023}. However, the notation `$(f\!:d_d,d_f)$-FCC' offers a more precise representation in terms of minimum distance, as illustrated in the following example.

\begin{example}
    Consider a function $f : \mathbb{F}_2^2 \rightarrow \{0,1\}$ defined by
$
f(00) = 0, \quad f(01) = f(10) = f(11) = 1.
$
For $t_f = 1$, both of the following codes achieve the same level of functional error correction:
$$
\{0000,\ 0111,\ 1011,\ 1111\}
\quad \text{and} \quad
\{0000,\ 0111,\ 1011,\ 1101\}.
$$
However, the first code has a minimum distance of 1, while the second code has a minimum distance of 2, allowing for single-error detection in the data.

The second code can thus be more precisely described using the notation $(f: d_d = 2, d_f = 3)$-FCC, which captures both the overall minimum distance and the minimum distance between codewords corresponding to different function values. In contrast, the notation $(f, t_d = 0, t_f = 1)$ fails to reflect this additional error-detection capability.
\end{example}

In case $d_d=d_f$, then an $(f\!:d_d,d_f)$-FCC is simply an error-correcting code with minimum distance at least $d_d$. We now present an example that illustrates our generalized framework of function-correcting codes.

\begin{example}\label{ex5}
Consider a function $f: \mathbb{F}_2^3 \rightarrow \{0,1,2,3\}$, where $f(u)$ denotes the position of the least frequent bit in the binary vector $u \in \mathbb{F}_2^3$, i.e.,
$f(000)=f(111)=0,  f(100)= f(011)=1,$
$f(010) = f(101)=2, f(001)=f(110)=3.$
For $t_f=2$, the function distance matrix is 
\[
\mathcal{D}_f (t_f, f_1,f_2,f_3,f_4)=\begin{bmatrix}
0 & 4 & 4 & 4 \\
4 & 0 & 4 & 4 \\
4 & 4 & 0 & 4\\
4 & 4 & 4 &0
\end{bmatrix},
\]
and $N(\mathcal{D}_f (t_f, f_1,f_2,f_3,f_4))=6$ determined using a trial-and-error approach. For vectors $u_1=000, u_2=100, u_3=010, u_4=001$, we have the following distance requirement matrix 
\[
\mathcal{D}_f (t_f, u_1,u_2,u_3,u_4)=\begin{bmatrix}
0 & 4 & 4 & 4 \\
4 & 0 & 3 & 3 \\
4 & 3 & 0 & 3\\
4 & 3 & 3 &0
\end{bmatrix}=\mathcal{D}.
\]
For this case also, we have $N(\mathcal{D}_f(t_f, u_1, u_2, u_3, u_4)) = 6$. Therefore, we have optimal redundancy $r_f(k, t_f)=6$ as 
$N(\mathcal{D}_f (t_f, u_1,u_2,u_3,u_4)) \le r_f(k,t_f) \le N(\mathcal{D}_f (t_f, f_1,f_2,f_3,f_4)).$

The following code is a $\mathcal{D}$-code for this function $f$, where one vector from the code is added as a parity to two message vectors that share the same function value:
$$
\{000000, 111100, 110011, 001111\}.
$$
As a result, the codewords of this FCC have length $9$ and are given by:\\
 $000000000, 111000000, 100111100, 011111100, 010110011,$ $ 101110011,$ $ 001001111, 110001111.$

It can be verified that the minimum distance of this code is $d = 3$. This implies that the code can correct up to $t_d = 1$ error in the data and up to $t_f = 2$ errors in the function values.
\end{example}

\subsection{\textbf{A general method for constructing $(f\!:d_d,d_f)$-FCCs }}\label{subsec:Method}

\textbf{Construction Method:}  
Given a function $f:\mathbb{F}_q^k \rightarrow \mathrm{Im}(f)$, the goal is to construct a code that provides protection against up to $t_d=\left\lfloor \frac{d_d-1}{2} \right\rfloor$ errors in the data and up to $t_f=\lfloor \frac{d_f-1}{2} \rfloor$ errors in the function values, where $t_d < t_f$, since protection of the data inherently contributes to the protection of the function values.

\begin{itemize}
    \item \textbf{Step 1:} Select an $[n, k, d_d]$ linear error-correcting code $\mathcal{C}$ with generator matrix $G$ of size $k \times n$. For any $u \in \mathbb{F}_q^k$, the corresponding codeword is given by $c_u = uG$.

    \item \textbf{Step 2:} Construct a function-correcting code (FCC) based on the set $\{c_u \mid u \in \mathbb{F}_q^k\}$, rather than directly on the message vectors $u$. That is, define a systematic encoding $\mathfrak{C}_f': \mathcal{C} \rightarrow \mathbb{F}_q^{n + r'}$ such that for any $u_1, u_2 \in \mathbb{F}_q^k$ with $f(u_1) \neq f(u_2)$, the following condition holds:
    \[
    d(\mathfrak{C}_f'(c_{u_1}), \mathfrak{C}_f'(c_{u_2})) \geq d_f.
    \]
\end{itemize}

Then, the resulting mapping $\mathfrak{C}_f: \mathbb{F}_q^k \rightarrow \mathbb{F}_q^{n + r'}$ defined by $\mathfrak{C}_f(u) = \mathfrak{C}_f'(c_u)$ is an $(f\!:d_d,d_f)$-FCC with total redundancy $r_s = n - k + r'$.

It is straightforward to verify that the encoding $\mathfrak{C}_f$ described above satisfies the properties of an $(f\!:d_d,d_f)$-FCC. Since $\mathcal{C}$ is an $[n,k,d_d]$ linear code, for any $u_1, u_2 \in \mathbb{F}_q^k$, such that $u_1 \ne u_2$, we have
$$d(\mathfrak{C}_f(u_1), \mathfrak{C}_f(u_2)) \geq d(c_{u_1}, c_{u_2}) \geq d_d.$$
For any $u_1, u_2 \in \mathbb{F}_q^k$ such that $f(u_1) \neq f(u_2)$, we have
$$d(\mathfrak{C}_f(u_1), \mathfrak{C}_f(u_2))=d(\mathfrak{C}_f'(c_{u_1}), \mathfrak{C}_f'(c_{u_2})) \geq d_f.$$

\textbf{Remark.} In Step 1 of the above construction, any $(n,q^k, 2t_d+1)$ non-linear code could also be used. However, due to the well-known advantages of linear codes such as compact representation and ease of encoding, we consider only linear codes.

We now introduce an extension of the notion of  DRM and FDM, where the distances are computed not over the space $\mathbb{F}_q^k$, but specifically over the $q^k$ codewords generated by a linear code of length $n$. This restriction allows us to analyze function-correcting properties within the structure of a given code, rather than across the full vector space $\mathbb{F}_q^k$. These extended definitions will help in finding suitable encodings in Step 2 of our construction method by using existing results on FCCs. In all the definitions given next, we consider a linear code $\mathcal{C}$ with generator matrix $G$ of order $k\times n$,  and $c_u=uG$ for any $u\in \mathbb{F}_q^k$.

\begin{definition}[Coded distance requirement matrix]\label{DCtf}
Let $u_1, u_2,$ $\ldots, u_{M} \in \mathbb{F}_q^k$, and $\mathcal{C}$ be a linear code with generator matrix $G$ of order $k\times n$. The coded distance requirement matrix (CDRM) $\mathcal{D}_{\mathcal{C},f}(t_f: u_1, u_2,\ldots, u_M)$ for an $(f, t_d, t_f)$-FCC is an $M \times M$ matrix with entries
$$
[\mathcal{D}_{\mathcal{C},f} (t_f: u_1, \ldots, u_M)]_{i, j} 
= 
\begin{cases}
\max (2t_f+1-d(c_{u_i}, c_{u_j}), 0), & \text{if } f(u_i) \neq f(u_j), \\
0, & \text{otherwise},
\end{cases}
$$
where $i, j \in \{1,2, \ldots, M\}$ and $c_u=uG$ for any $u\in \mathbb{F}_q^k$.
\end{definition}

Note that in the trivial case, when $\mathcal{I}$ denotes the identity map on $\mathbb{F}_q^k$ or the code whose generator matrix is an identity matrix of order $k$, we have
%$t_d=0$, then $n=k$ and $d(c_{u}, c_{v})=d(u,v)$. Therefore, we have $\mathcal{D}_{\mathcal{C},f} (t_f: u_1, \ldots, u_M)=\mathcal{D}_{f} (t_f, u_1, \ldots, u_M)$. In other words,
$$\mathcal{D}_{\mathcal{I},f} (t_f: u_1, \ldots, u_M)=\mathcal{D}_{f} (t_f, u_1, \ldots, u_M).$$

\begin{definition}
For a  function $f: \mathbb{F}_q^k \to \mathrm{Im}(f)$, the coded distance between $f_i, f_j \in \mathrm{Im}(f)$ is defined as
$$d_{\mathcal{C}}(f_i, f_j) = \min_{u_1, u_2 \in \mathbb{F}_q^k} \{d(c_{u_1}, c_{u_2}) |  f(u_1)=f_i, f(u_2) = f_j\}.$$
\end{definition}

\begin{definition}[Coded function distance matrix]
Consider a function $f: \mathbb{F}_q^k \to \mathrm{Im}(f)$ and $E=|\mathrm{Im}(f)|$. Then $E \times E$ matrix $\mathcal{D}_{\mathcal{C},f}(t_f: f_1, f_2,\ldots, f_E)$ with entries given as
$$
[\mathcal{D}_{\mathcal{C},f}(t_f: f_1, f_2,\ldots, f_E)]_{i, j} 
= 
\begin{cases}
\max(2t_f+1 -d_{\mathcal{C}}(f_i, f_j), 0), & \text{if } i \neq j, \\
0, & \text{otherwise},
\end{cases}
$$
is called a coded function distance matrix (CFDM).
\end{definition}

Using our construction method and the extended definitions of distance matrices, we construct an $(f, t_d, t_f)$-FCC in the following example. This example illustrates our construction procedure and demonstrates that an $(f, t_f)$-FCC does not necessarily guarantee data protection. In contrast, the code generated using our method, with the same code length, ensures both function value and data protection.

\begin{example}\label{ex2}
Consider Hamming weight function $f: \mathbb{F}_2^3 \rightarrow \{0,1,2,3\}$ defined as $f(u)=$wt$(u)$, where $u \in  \mathbb{F}_2^3$, i.e.,
\begin{align*}
f(000)&=0,\\
f(100)&= f(010)=f(001)=1, \\
f(110) &= f(101)=f(011)=2,\\
f(111)&=3.
\end{align*}
%$$f(000)=0, f(100)= f(010)=f(001)=1, $$
%$$f(110) = f(101)=f(011)=2, f(111)=3.$$

The aim is to construct a code that provides protection from up to $t_d=1$ error in the data and up to $t_f=2$ errors in function values.

\begin{itemize}
\item \textbf{Step 1:} Select the $[6, 3, 3]$ linear error correcting code $\mathcal{C}$ with generator matrix 
$$G=\begin{bmatrix}
1 & 0 & 0 & 1 & 1& 0 \\
0 & 1 & 0 & 1 & 0& 1 \\
0 & 0 & 1 & 0 & 1& 1 \\
\end{bmatrix}.
$$ 
For any $u\in \mathbb{F}_2^k$, we denote the corresponding codeword as $c_u=uG$.

\item \textbf{Step 2:} Find an FCC based on $\{c_u \mid u \in \mathbb{F}_2^3\}$ not based on $u\in \mathbb{F}_2^3$. For vectors $u_1=000, u_2=100, u_3=011, u_4=111$, we have $c_{u_1}=000000, c_{u_2}=100110, c_{u_3}=011110, c_{u_4}=111000$. Therefore, for $t_f=2$, the corresponding CDRM is
\[
\mathcal{D}_{\mathcal{C},f} (t_f: u_1,u_2,u_3,u_4)=\begin{bmatrix}
0 & 2 & 1 & 2 \\
2 & 0 & 2 & 1 \\
1 & 2 & 0 & 2\\
2 & 1 & 2 &0
\end{bmatrix}=\mathcal{D},
\]
for which, we have $N(\mathcal{D}_{\mathcal{C},f}  (t_f\!:\! u_1,u_2,u_3,u_4))=3$ and a $\mathcal{D}$-code achieving this value is $\{000,110,101,011\}$.
Further, the coded function distance matrix  $\mathcal{D}_{\mathcal{C},f} (t_f: f_1,f_2,f_3,f_4)$ is same as CDRM given above. Therefore, the optimal redundancy for systematic encoding $\mathfrak{C}_f'$ is $r'=3$. 
We use a $\mathcal{D}$-code with optimal redundancy $3$, given by $$\{000,110,101,011\},$$ by adding one vector from this set as the parity to all message vectors having the same function value.
\end{itemize}
Hence, our code is defined as 
\[
\begin{matrix}
\text{Message}  & \text{$\mathcal{C}$-parity} \ \mathcal{C} & \text{FCC-parity} & \text{Codeword} \\
000 & 000   & 000 & 000000000 \\
100 & 110   & 110 & 100110110 \\
010 & 101   & 110 & 010101110 \\
001 & 011   & 110 & 001011110 \\
110 & 011   & 101 & 110011101 \\
101 & 101   & 101 & 101101101 \\
011 & 110   & 101 & 011110101 \\
111 & 000   & 011 & 111000011 
\end{matrix}
\]

On the other hand, if we construct an $(f, t_f)$-FCC for this function $f$ with $t_f = 2$ in the traditional way, the optimal redundancy is 6, as confirmed by the lower bound in Corollary \ref{colHb}, resulting in a code length equal to that of our proposed code. However, in that case, any protection at the data level is not guaranteed. To demonstrate this, we present the following $(f, 2)$-FCC.
\[
\begin{matrix}
\text{Message}  &  \text{Redundancy in FCC} & \text{Codeword} \\
000 & 000000 & 000000000 \\
100 & 111100 & 100111100 \\
010 & 111100 & 010111100 \\
001 & 111100 & 001111100 \\
110 & 110011 & 110110011 \\
101 & 110011 & 101110011 \\
011 & 110011 & 011110011 \\
111 & 001111 & 111001111
\end{matrix}
\]

As can be checked, the minimum distance of this code is $d=2$, which means it cannot correct one error in the data.

\end{example}
%\clearpage

\section{Bounds for optimal redundancy for FCCs with Data Protection}\label{general_definitions}

In this section, we first extend the definitions introduced in \cite{LBWY2023} for $(f, t)$-FCCs to more general $(f, t_d, t_f)$-FCCs. We then derive bounds on the redundancy of these codes using irregular-distance codes.

\begin{definition}[Optimal redundancy]
The \emph{optimal redundancy}, denoted by $r_f(k :\! d_d, d_f)$ or equivalently by $r_f(k, t_d, t_f)$, is defined as the minimum value of $r$ for which there exists an $(f\!:d_d,d_f)$-FCC with an encoding function $\mathfrak{C}_f: \mathbb{F}_q^k \rightarrow \mathbb{F}_q^{k + r}$. 
\end{definition}

For a function $f: \mathbb{F}_q^k \rightarrow \mathrm{Im}(f)$, any $(f, t_d, t_f)$-FCC is an $(f, t_f)$-FCC, and not conversely. Therefore, it is clear from the definition of optimal redundancy that 
$$r_f(k,t_f)\le r_f(k, t_d, t_f).$$
For $t_d=0$, we have $r_f(k, t_d, t_f)=r_f(k, t_f)$. For some function, equality might hold for some non-zero $t_d$, such as in Example \ref{ex2}. In Example \ref{ex2}, for $t_d=1$, we have 
$$r_f(k, t_d, t_f)=r_f(k, t_f)=6.$$

\begin{definition}\label{Dtdtf}
Let $u_1, u_2,$ $\ldots, u_{M} \in \mathbb{F}_q^k$. The distance requirement matrix (DRM) $\mathcal{D}_{f}(t_d,t_f: u_1, u_2,\ldots, u_M)$ for an $(f, t_d, t_f)$-FCC is an $M \times M$ matrix with entries
$$
[\mathcal{D}_{f} (t_d, t_f: u_1, \ldots, u_M)]_{i, j} 
= 
\begin{cases}
\max (2t_d+1-d(u_i, u_j), 0), & u_i \ne u_j, f(u_i) = f(u_j), \\
\max (2t_f+1-d(u_i, u_j), 0), & f(u_i) \neq f(u_j), \\
0, & \text{otherwise},
\end{cases}
$$
where $i, j \in \{1,2, \ldots, M\}$.
\end{definition}

The representation of the distance requirement matrix always depends on the order in which the vectors in $\mathbb{F}_q^k$ are taken. However, regardless of the order, the matrices are equivalent up to row and column permutations.

\begin{example}\label{ex1.2}
\normalfont Consider the Hamming weight function $f: \mathbb{F}_2^3 \to \{0, 1, 2, 3\}$ where $f(u)=\text{wt}(u)$ for all $u\in \mathbb{F}_2^3$.
Then for $t_d=1$ and $t_f=2$, the distance requirement matrix is
$$\mathcal{D}_f(t_d,t_f: u_1, u_2, \ldots, u_8) =
\begin{bmatrix} 
0 & \textcolor{red}{\mathbf{4}} & \textcolor{red}{\mathbf{4}} & \textcolor{red}{\mathbf{4}} & \textcolor{red}{\mathbf{3}} & \textcolor{red}{\mathbf{3}} & \textcolor{red}{\mathbf{3}} & \textcolor{red}{\mathbf{2}}\\
\textcolor{red}{\mathbf{4}} & 0 & \textcolor{blue}{\underline{1}} & \textcolor{blue}{\underline{1}} & \textcolor{red}{\mathbf{4}} & \textcolor{red}{\mathbf{4}} & \textcolor{red}{\mathbf{2}} & \textcolor{red}{\mathbf{3}}\\
\textcolor{red}{\mathbf{4}} & \textcolor{blue}{\underline{1}} & 0 & \textcolor{blue}{\underline{1}} & \textcolor{red}{\mathbf{4}} & \textcolor{red}{\mathbf{2}} & \textcolor{red}{\mathbf{4}} & \textcolor{red}{\mathbf{3}}\\
\textcolor{red}{\mathbf{4}} & \textcolor{blue}{\underline{1}} & \textcolor{blue}{\underline{1}} & 0 & \textcolor{red}{\mathbf{2}} & \textcolor{red}{\mathbf{4}} & \textcolor{red}{\mathbf{4}} & \textcolor{red}{\mathbf{3}}\\
\textcolor{red}{\mathbf{3}} & \textcolor{red}{\mathbf{4}} & \textcolor{red}{\mathbf{4}} & \textcolor{red}{\mathbf{2}} & 0 & \textcolor{blue}{\underline{1}} & \textcolor{blue}{\underline{1}} & \textcolor{red}{\mathbf{4}}\\
\textcolor{red}{\mathbf{3}} & \textcolor{red}{\mathbf{4}} & \textcolor{red}{\mathbf{2}} & \textcolor{red}{\mathbf{4}} & \textcolor{blue}{\underline{1}} & 0 & \textcolor{blue}{\underline{1}} & \textcolor{red}{\mathbf{4}}\\
\textcolor{red}{\mathbf{3}} & \textcolor{red}{\mathbf{2}} & \textcolor{red}{\mathbf{4}} & \textcolor{red}{\mathbf{4}} & \textcolor{blue}{\underline{1}} & \textcolor{blue}{\underline{1}} & 0 & \textcolor{red}{\mathbf{4}}\\
\textcolor{red}{\mathbf{2}} & \textcolor{red}{\mathbf{3}} & \textcolor{red}{\mathbf{3}} & \textcolor{red}{\mathbf{3}} & \textcolor{red}{\mathbf{4}} & \textcolor{red}{\mathbf{4}} & \textcolor{red}{\mathbf{4}} & 0
\end{bmatrix},$$
%\begin{bmatrix} 
%0 & 4 & 4 & 4 & 3 & 3 & 3 & 2\\
%4 & 0 & 1 & 1 & 4 & 4 & 2 & 3\\
%4 & 1 & 0 & 1 & 4 & 2 & 4 & 3\\
%4 & 1 & 1 & 0 & 2 & 4 & 4 & 3\\
%3 & 4 & 4 & 2 & 0 & 1 & 1 & 4\\
%3 & 4 & 2 & 4 & 1 & 0 & 1 & 4\\
%3 & 2 & 4 & 4 & 1 & 1 & 0 & 4\\
%2 & 3 & 3 & 3 & 4 & 4 & 4 & 0
%\end{bmatrix}.$$
where the information vectors are considered in the following order $000, 100, 010, 001, 110, 101, 011, 111$. The red (bold) entries  correspond to vector pairs with different function values, while the blue (underlined) entries correspond to vector pairs with the same function value.
\end{example}

%\noindent \textbf{Remark:} 
For $\mathcal{D}=\mathcal{D}_f(t_d, t_f: u_1, u_2,\ldots, u_{2^k})$, if we have a $\mathcal{D}$-code $\mathcal{P}=\{p_1, p_2, \ldots, p_{2^k}\}$ then it can be used to construct a $(f, t_d, t_f)$-FCC with the encoding $\mathcal{C}(u_i) = (u_i, p_i)$ for all $i \in \{1, 2, \ldots, 2^k\}$.

\begin{example}
\normalfont Consider the same function  $f: \mathbb{F}_2^3\to \{0, 1,2,3\}$ from Examples \ref{ex2} and \ref{ex1.2}. As given in Example \ref{ex2}, we have a $\mathcal{D}$-code, $\mathcal{P}=\{000000, 110110, 101110, 011110, 011101, 101101,$ $110101, 000011\}$ for which the distance distribution structure is
\begin{equation*}
\resizebox{\columnwidth}{!}{
 \begin{blockarray}{ccccccccc}
   &  000000 & 110110 & 101110 & 011110 & 011101 & 101101 & 110101 & 000011 \\
  \begin{block}{c[cccccccc]}
000000 & 0 & 4 & 4 & 4 & 4 & 4 & 4 & 2\\
110110 & 4 & 0 & 2 & 2 & 4 & 4 & 2 & 4\\
101110 & 4 & 2 & 0 & 2 & 4 & 2 & 4 & 4\\
011110 & 4 & 2 & 2 & 0 & 2 & 4 & 4 & 4\\
011101 & 4 & 4 & 4 & 2 & 0 & 2 & 2 & 4\\
101101 & 4 & 4 & 2 & 4 & 2 & 0 & 2 & 4\\
110101 & 4 & 2 & 4 & 4 & 2 & 2 & 0 & 4\\
000011 & 2 & 4 & 4 & 4 & 4 & 4 & 4 & 0 \\
\end{block}
\end{blockarray}.}
\end{equation*}

From Corollary~\ref{colHb}, we know that $r_f(k, t_f = 2) \geq 5.25$. Additionally, the existence of an $(f, t_f)$-FCC with redundancy 6 has already been demonstrated in Example~\ref{ex2}. Therefore, we conclude that $r_f(k = 3, t_f = 2) = 6$. Since $r_f(k, t_f) \leq r_f(k, t_d, t_f)$, it follows that 6 is the minimum possible length for a $\mathcal{D}$-code, and we have $r_f(k, t_d, t_f)=6$.
\end{example}

\begin{theorem}\label{thm:r=ND}
    For any function $f:\mathbb{F}_q^k \rightarrow \mathrm{Im}(f)$, we have
    $$r_f(k,t_d,t_f) = N(\mathcal{D}_{f} (t_d, t_f: u_1, \ldots, u_{q^k})).$$
\end{theorem}
\begin{proof}
   Let $\mathcal{D}_{f} (t_d, t_f: u_1, \ldots, u_{q^k})$ be denoted by $\mathcal{D}$ for notational simplicity. Clearly, $r_f(k,t_d,t_f) \leq N(\mathcal{D})$ as a  $\mathcal{D}$-code of length $N(\mathcal{D})$ can be used as redundancy for an $(f,t_d,t_f)$-FCC. For the converse part assume that $r_f(k,t_d,t_f) < N(\mathcal{D})$. This means there exists an $(f, t_d,t_f)$-FCC $\mathfrak{C}_f$ with redundancy $r_f(k,t_d,t_f)$. Then at least one of the following cases is true.
   \begin{enumerate}
       \item There exist two redundancy vectors $p_u,p_v \in \mathbb{F}_q^{r_f(k,t_d,t_f)}$ for some $u,v \in \mathbb{F}_q^k, u\ne v$ such that $d(p_u, p_v) < 2t_d+1 -d(u,v)$.
       \item There exist two redundancy vectors $p_{u'},p_{v'} \in \mathbb{F}_q^{r_f(k,t_d,t_f)}$ for some $u',v' \in \mathbb{F}_q^k, f(u')\ne f(v')$ such that $d(p_{u'}, p_{v'}) < 2t_f+1 -d(u',v')$.
   \end{enumerate}
  
   That means $d(\mathfrak{C}_f(u), \mathfrak{C}_f(v)) < 2t_d+1$ or $d(\mathfrak{C}_f(u'), \mathfrak{C}_f(v')) < 2t_f+1$, respectively, which is a contradiction to the definition of an $(f, t_d, t_f)$-FCC. Therefore,  $r_f(k,t_d,t_f) = N(D)$.
\end{proof}

\begin{theorem}\label{thm1.1}
For any function $f: \mathbb{F}_q^k \to \mathrm{Im}(f)$ and $\{u_1, u_2, \ldots, u_m\}\subseteq \mathbb{F}_q^k$, we have
$$r_f(k, t_d, t_f) \geq N(\mathcal{D}_f(t_d, t_f: u_1, u_2, \ldots, u_m)),$$
and  $r_f (k, t_d, t_f) \geq 2t_f$ for $|\mathrm{Im}(f)|\geq 2$.

Further, $r_f (k, t_d, t_f) \geq N(q^k,2t_d+1)-k$, where $N(M, d)$ is the minimum length of an error-correcting code with $M$ codewords and minimum distance $d$.
\end{theorem}
\begin{proof}
    The first part is straightforward, as an optimal FCC should satisfy the distance requirement for all the vectors in $\mathbb{F}_q^k$. 
 If $|\mathrm{Im}(f)|\geq 2$, then there exist $u,v \in \mathbb{F}_q^k$ such that $f(u) \ne f(v)$ and $d(u,v)=1$.Therefore, $$r_f(k, t_d, t_f) \geq N(\mathcal{D}_f(t_d, t_f: u, v))=2t_f$$ using repetition code.
 Furthermore, by the definition of $\mathcal{D}_f(t_d, t_f: u_1, u_2, \ldots, u_{q^k})$, for any $i,j \in [q^k]$,
$$[\mathcal{D}_{f} (t_d, t_f: u_1, \ldots, u_{q^k})]_{i, j}\ge \max (2t_d+1-d(u_i, u_j), 0),$$ 
whenever $u_i\ne u_j$ as $t_f\ge t_d$. Since this is the requirement for any error-correcting code with minimum distance at least $2t_d+1$, we have $r_f (k, t_d, t_f) \geq N(q^k,2t_d+1)-k$.
 
\end{proof}

For the binary vector space $\mathbb{F}_2^k$, we derive a tighter lower bound on the redundancy of FCC, as given below.

\begin{theorem}\label{thm1.2}
For any function $f: \mathbb{F}_2^k \to \mathrm{Im}(f)$, if $|\mathrm{Im}(f)| \geq 2$ and $k\geq 2$, then $$r_f (k, t_d, t_f) \geq 2t_f+t_d.$$
\end{theorem}
\begin{proof}
 Since $|\mathrm{Im}(f)| \geq 2$, there exist two vectors $u_1, u_2 \in \mathbb{F}_2^k$ such that $d(u_1, u_2) = 1$ and $f(u_1) \neq f(u_2)$. Now consider another vector $u_3=u_1+e_i,$
where $e_i$ is the $i$-th standard basis binary vector of length $k$, and $u_3\ne u_2$. Now we have following 3 cases with $f(u_3)$:\\
\textbf{Case 1:} If $f(u_3)=f(u_2)$, then we have $d(u_1,u_2)=d(u_1, u_3)=1, d(u_2,u_3)=2$ and $f(u_1)\neq f(u_2)=f(u_3)$. Therefore,
\[
     \mathcal{D}_f(t_d, t_f:u_1,u_2,u_3)= 
     \begin{bmatrix}
         0 & 2t_f & 2t_f \\
         2t_f & 0 & 2t_d-1 \\
         2t_f & 2t_d-1 & 0
     \end{bmatrix}.
     \]   
\textbf{Case 2:} If $f(u_3)= f(u_1)$, then we have $d(u_1, u_2) = d(u_1, u_3) = 1$, $d(u_2, u_3) = 2$, and $f(u_1) = f(u_3) \ne f(u_2)$. Therefore,
\[
     \mathcal{D}_f(t_d, t_f:u_1,u_2,u_3)= 
     \begin{bmatrix}
         0 & 2t_f & 2t_d \\
         2t_f & 0 & 2t_f-1 \\
         2t_d & 2t_f-1 & 0
     \end{bmatrix}.
     \]
 \textbf{Case 3:} If $f(u_3)\neq f(u_1)$ and $f(u_3) \neq f(u_2)$, then 
 we have $d(u_1, u_2) = d(u_1, u_3) = 1$, $d(u_2, u_3) = 2$, and $f(u_1) \neq f(u_2) \ne f(u_3)$. Therefore,
\[
     \mathcal{D}_f(t_d, t_f:u_1,u_2,u_3)= 
     \begin{bmatrix}
         0 & 2t_f & 2t_f \\
         2t_f & 0 & 2t_f-1 \\
         2t_f & 2t_f-1 & 0
     \end{bmatrix}.
     \]
    Now consider the matrix $\mathcal{D}_f(t_d, t_f:u_1,u_2,u_3)=\mathcal{D}$ from Case 1.  Let $\{p_1, p_2, p_3\}$ be a $\mathcal{D}$-code of length $r$, so we have $d(p_1, p_2) \ge 2t_f$, $ d(p_1, p_3) \geq 2t_f$ and $d(p_2, p_3) \geq 2t_d - 1$. Further, for any three binary vectors of length $r$, we have  $d(p_1, p_2) + d(p_1, p_3) + d(p_2, p_3) \leq 2r$. Combining these, we obtain 
    $$4t_f \leq d(p_1, p_2) + d(p_1, p_3) \leq 2r - (2t_d - 1),$$ which implies $r \geq \frac{4t_f + 2t_d-1}{2} = 2t_f + t_d - \frac{1}{2}$. Therefore, we conclude that $r \geq 2t_f + t_d$, and hence $r_f(k, t_d, t_f) \geq 2t_f + t_d$. A similar proof will follow for matrix $\mathcal{D}_f(t_d, t_f:u_1,u_2,u_3)$ in Case 2. Applying the same logic to the matrix $\mathcal{D}_f(t_d, t_f:u_1,u_2,u_3)$ in Case 3, we get 
    $$r\geq \frac{6t_f-1}{2} \ge 3t_f-\frac{1}{2}.$$
Since $r$ is an integer and $t_f\geq t_d$, we have $r\geq 3t_f\ge 2t_f+t_d$.
\end{proof}

\subsection{Upper bounds for $r_f(k,t_d,t_f)$ }

We have defined two types of distance requirement matrices in Definitions \ref{DCtf} and \ref{Dtdtf}, given as
\begin{align*}
&[\mathcal{D}_{\mathcal{C},f} (t_f: u_1, \ldots, u_M)]_{i, j} = 
\begin{cases}
\max (2t_f+1-d(c_{u_i}, c_{u_j}), 0), & \text{if } f(u_i) \neq f(u_j), \\
0, & \text{otherwise},
\end{cases} \\[1em]
&[\mathcal{D}_{f} (t_d, t_f: u_1, \ldots, u_M)]_{i, j} = 
\begin{cases}
\max (2t_d+1-d(u_i, u_j), 0), & \text{if } u_i \ne u_j, f(u_i) = f(u_j), \\
\max (2t_f+1-d(u_i, u_j), 0), & \text{if } f(u_i) \neq f(u_j), \\
0, & \text{otherwise}.
\end{cases}
\end{align*}

In the following theorem, we make a connection between these two matrices, and using that, we provide an upper bound for $r_f(k,t_d,t_f)$.

\begin{theorem}\label{thm:Dcode}
   Let $\mathcal{C}$ be an $[n,k, 2t_d+1]$ error-correcting code, and let $c_u$ denote the codeword that corresponds to the message vector $u\in \mathbb{F}_q^k$. For any function $f:\mathbb{F}_q^k \rightarrow \mathrm{Im}(f)$,
    $$ N(\mathcal{D}_{f} (t_d, t_f: u_1, \ldots, u_{q^k})) \le N(\mathcal{D}_{\mathcal{C},f} (t_f: u_1, \ldots, u_{q^k})) +n-k.$$
\end{theorem}

\begin{proof}
    Without loss of generality, consider a systematic form of $\mathcal{C}$, i.e., $c_u=(u,w_u)$ for all $u\in \mathbb{F}_q^k$. For ease of notation, we use 
$\mathcal{D}_1=\mathcal{D}_{f} (t_d, t_f: u_1, \ldots, u_{q^k})$ and $\mathcal{D}_2=\mathcal{D}_{\mathcal{C},f} (t_f: u_1, \ldots, u_{q^k})$. Let $\mathcal{P}_2=\{p_u \mid u \in\mathbb{F}_q^k\}$ be an ordered $\mathcal{D}_2$-code with length $N(\mathcal{D}_2)$. Now, we define a code as 
$$\mathcal{P}_1=\{ \overline{p}_u=(w_u,p_u) \mid u \in\mathbb{F}_q^k\},$$
and show that this is a $\mathcal{D}_1$-code, for which we need to show that $d(\overline{p}_u, \overline{p}_v) \ge [\mathcal{D}_1]_{u,v}$, where $u,v \in \mathbb{F}_q^k$ and the rows and columns of matrix $\mathcal{D}_1$ are indexed by the vectors in $\mathbb{F}_q^k$. For any $u,v \in \mathbb{F}_q^k$,  
\begin{equation}\label{eq:Dcode}
    d(\overline{p}_u, \overline{p}_v) = d(w_u, w_v)+d(p_u,p_v).
\end{equation}
%If $u \ne v$ and $f(u) = f(v)$, then $d(p_u,p_v)=0$, and from \eqref{eq:Dcode}, we get 
If $u \ne v$ and $f(u) = f(v)$, then $d(p_u,p_v)\geq 0$, and from \eqref{eq:Dcode}, we get 
\begin{align*}
    d(\overline{p}_u, \overline{p}_v) &\geq d(w_u,w_v)=d(c_u,c_v)-d(u,v) \\ 
    &\ge 2t_d+1-d(u,v)=[\mathcal{D}_1]_{u,v}.
\end{align*}
If $f(u) \ne f(v)$, then by the definition of $\mathcal{D}_2$-code and \eqref{eq:Dcode}, 
\begin{align*}
    d(\overline{p}_u, \overline{p}_v) &\ge d(w_u,w_v) + 2t_f+1-d(c_u, c_v)\\ 
    &=2t_f+1-(d(c_u, c_v)-d(w_u,w_v)) \\
    &= 2t_f+1-d(u,v)=[\mathcal{D}_1]_{u,v}.
\end{align*}
Hence, $\mathcal{P}_1$ is a $\mathcal{D}_1$-code with length $n-k+N(\mathcal{D}_2)$, and we have $N(\mathcal{D}_1) \le n-k+N(\mathcal{D}_2)$
\end{proof}

The following corollary directly holds by Theorem \ref{thm:r=ND} and \ref{thm:Dcode}, and gives an upper bound on the optimal redundancy of an $(f, t_d, t_f)$-FCC.

\begin{corollary}\label{col:r<ND}
       For any function $f:\mathbb{F}_q^k \rightarrow \mathrm{Im}(f)$,
    $$r_f(k,t_d,t_f) \le N(\mathcal{D}_{\mathcal{C},f} (t_f: u_1, \ldots, u_{q^k})) +n-k,$$
where $\mathcal{C}$ is an $[n,k,2t_d+1]$ error-correcting code.
\end{corollary}

Since the matrix $\mathcal{D}_{\mathcal{C},f}(t_f\!:\! u_1, \ldots, u_{q^k})$ has a similar form as the matrix $\mathcal{D}_{f}(t_f, u_1, \ldots, u_{q^k})$ in \cite{LBWY2023}, all the constructions given for FCC in the literature can be used to find $N(\mathcal{D}_{\mathcal{C},f}(t_f: u_1, \ldots, u_{q^k})).$

From \cite[Appendix A]{LBWY2023}, there exists a binary code of length $n$ and minimum distance $2t+1$ and redundancy at most $$r\le \frac{t\log k+t}{(1-\frac{t}{k}\log e)}.$$ Therefore, we have the following corollary.

\begin{corollary}\label{col2:r<ND}
       For any function $f:\mathbb{F}_2^k \rightarrow \mathrm{Im}(f)$, we have 
    $$r_f(k,t_d,t_f) \le N(\mathcal{D}_{\mathcal{C},f} (t_f: u_1, \ldots, u_{q^k})) +\frac{t_d\log k+t_d}{(1-\frac{t_d}{k}\log e)}.$$
where $\mathcal{C}$ is an $[n,k, 2t_d+1]$ binary error-correcting code, and $n\le k+t_d\lceil \log n \rceil $.
\end{corollary}

\begin{theorem}
Let $\mathcal{C}$ be an $[n,k,2t_d+1]$ code. Then 
$$N(\mathcal{D}_{\mathcal{C},f}(t_f: u_1, u_2,\ldots, u_M)) \le N(M, 2(t_f-t_d)).$$
\end{theorem}
\begin{proof}
    Since $d(c_{u_i}, c_{u_j}) \ge 2t_d+1$ for all $u_i, u_j \in \mathbb{F}_q^k$, we have 
    \begin{align*}
       [\mathcal{D}_{\mathcal{C},f}(t_f: u_1, u_2,\ldots, u_M)]_{i,j} &=  \max (2t_f+1-d(c_{u_i}, c_{u_j}), 0) \le 2(t_f-t_d). 
    \end{align*}
    Therefore, $N(\mathcal{D}_{\mathcal{C},f}(t_f: u_1, u_2,\ldots, u_M)) \le N(M, 2(t_f-t_d)).$
\end{proof}

Using Lemma \ref{lem:4from1} and the above theorem, we get the following corollary.

\begin{corollary}
    Let $\mathcal{C}$ be an $[n,k, 2t_d+1]$ code. Then for $t_f-t_d\geq 5$ and $M\leq 4(t_f-t_d)^2$, we have 
\begin{align*}
    N(\mathcal{D}_{\mathcal{C},f}(t_f: u_1,\ldots, u_M)) &\le N(M, 2(t_f-t_d)) \leq \frac{4(t_f-t_d)-2}{1-2\sqrt{\ln{(2(t_f-t_d))}/2(t_f-t_d)}}.
\end{align*}
Furthermore, for $M=4$, we have
$N(4, 2(t_f-t_d)) = 3(t_f-t_d).$
\end{corollary}

The last conclusion for $M=4$, namely $N(4,2(t_f-t_d))=3(t_f-t_d),$
is an immediate consequence of \cite[Lemma 3]{RRHH2025} which proves $N(4,2t)=3t$.

Using the lower and upper bounds on the redundancy of an $(f, t_f)$-FCC, provided in Corollary \ref{thm1} and Theorem \ref{thm2}, respectively, the following theorem establishes bounds on the redundancy for our method of constructing codes described in Subsection~\ref{subsec:Method}.

\begin{theorem}
Let $\mathcal{C}$ be an $[n, k, 2t_d + 1]$ linear code, where $n$ denotes the minimum possible length of a linear code with dimension $k$ and minimum distance at least $2t_d + 1$. Then, the redundancy $r_s$ of the proposed scheme is bounded as:
$$
N(\mathcal{D}_{\mathcal{C},f}(t_f: u_1,\ldots, u_M)) +n-k \\
\le r_s \le N(\mathcal{D}_{\mathcal{C},f}(t_f: f_1, \ldots, f_E))+n-k.
$$
\end{theorem}

\begin{remark}
The reduction of the redundancy problem to $N(D)$ is intended as a design framework rather than only as an existence statement. By Theorem~\ref{thm:r=ND}, the optimal redundancy is exactly the minimum length of a $D$-code for the associated DRM. In practice, one first forms the relevant DRM/FDM (or, in the two-step construction, the CDRM/CFDM), and then studies the corresponding $D$-code problem. When the function has additional structure, this matrix is often much smaller or possesses useful symmetry, which makes the computation or estimation of $N(D)$ more tractable. The systematic structure of the DRM was used in \cite{LBWY2023} to construct FCCs for the Hamming weight function.

For small instances, $N(D)$ can be obtained directly by search or explicit construction, as illustrated in Example~\ref{ex5}. For larger instances, the DRM corresponding to only a suitably chosen subset of vectors of $\mathbb{F}_q^k$, rather than all $q^k$ vectors, can be used to obtain lower bounds on the optimal redundancy, as described in Theorem~\ref{thm1.1}.
\end{remark}

\section{Non-existence of {\it strict  $(f\!:\!d_d,d_f)$-FCC}}\label{sec5}

We call an  $(f\!:\!d_d,d_f)$-FCC a {\it strict  $(f\!:\!d_d,d_f)$-FCC} if $d_f > d_d$ since the case $d_f=d_d$ reduces to ECC. In this section, we show that certain codes cannot be used as strict $(f\!:\!d_d,d_f)$-FCCs by analyzing their minimum-distance graph, defined as follows.

\subsection{Minimum-Distance Graphs}
\begin{definition}[Minimum-distance graph]
Let $\mathcal{C}$ be a code. The \emph{minimum-distance graph} of $\mathcal{C}$, denoted by $G(\mathcal{C})$, is the graph whose vertex set is $\mathcal{C}$ and two distinct vertices $c_1, c_2 \in \mathcal{C}$ are adjacent if and only if 
$d(c_1, c_2) = d_{\min}(\mathcal{C}),
$
where $d_{\min}(\mathcal{C})$ denotes the minimum distance of the code $\mathcal{C}$.
\end{definition}

\begin{example}
    Let
$\mathcal C=\{0000,\,0011,\,1100,\,1111\}\subset \mathbb{F}_2^4.
$
Then $d_{\min}(\mathcal C)=2$. The minimum-distance graph $G(\mathcal C)$ has edges only between
those pairs at distance $2$, forming a 4-cycle.

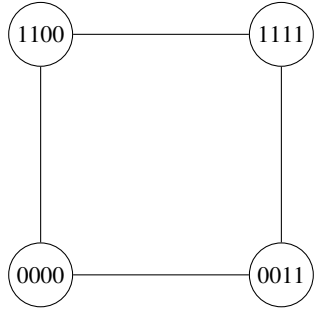
\begin{figure}[h]
\centering
\begin{tikzpicture}[scale=0.8,
  every node/.style={circle,draw,inner sep=2pt,minimum size=18pt}]
  \node (a) at (0,0)   {\small 0000};
  \node (b) at (4,0)   {\small 0011};
  \node (c) at (4,4)   {\small 1111};
  \node (d) at (0,4)   {\small 1100};

  % edges at distance 2
  \draw (a)--(b);
  \draw (a)--(d);
  \draw (c)--(b);
  \draw (c)--(d);
\end{tikzpicture}
%\caption{Minimum-distance graph \(G(\mathcal C)\) is a 4-cycle.}
\end{figure}
\end{example}

The following theorem illustrates how knowledge of the minimum-distance graph can be used to determine whether a code can serve as a strict $(f\!:d_d,d_f)$-FCC.

\begin{theorem}\label{MDG1}
Let $\mathcal{C}$ be a $(n,q^k,d)$ code. If the minimum-distance graph $G(\mathcal{C})$ is a connected graph, then $\mathcal{C}$ cannot be an $(f\!:d,d_f)$-FCC for any $f:\mathbb{F}_q^k\to \mathrm{Im}(f)$ with $|\mathrm{Im}(f)| \geq 2$ and $d_f>d$, (equivalently, $\mathcal{C}$ cannot be a strict $(f:d,d_f)$-FCC). 
\end{theorem}

\begin{proof}
Assume $G(\mathcal C)$ is connected and let $\mathcal C$ serves as an $(f\!:d,d_f)$ with $d_f>d$ for a function $f:\mathbb F_q^k\to \mathrm{Im}(f)$. For each $a\in\mathrm{Im}(f)$, define the following set
$$
\mathcal C_a=\{c_u\in\mathcal C :\ f(u)=a\},
$$
where $c_u$ denotes the codeword corresponding to message vector $u\in \mathbb{F}_q^k$. If $|\mathrm{Im}(f)|\ge2$ and $d_f>d$, then by the $(f\!:d,d_f)$-FCC requirement, every $x\in\mathcal C_a$ and $y\in\mathcal C_b$ with $a\ne b$ must satisfy $d(x,y)\ge d_f>d$. Hence, no edge of $G(\mathcal C)$ exists between distinct $\mathcal C_a$ and $\mathcal C_b$, so the partition $\{\mathcal C_a, \mathcal C_b\}$ separates the graph, contradicting the connectedness of $G(\mathcal C)$. Therefore, $\mathcal{C}$ cannot be an $(f\!:d,d_f)$-FCC.
\end{proof}

Theorem \ref{MDG1} can easily be generalized for codes whose minimum-distance graph has more than one connected component. 

\begin{theorem}\label{MDG2}
Let $\mathcal{C}$ be a $(n,q^k,d)$ code. If the minimum-distance graph $G(\mathcal{C})$ has $Q$ number of connected components, then $\mathcal{C}$ cannot be a $(f\!:d,d_f)$-FCC for any $f:\mathbb{F}_q^k\to \mathrm{Im}(f)$ with $|\mathrm{Im}(f)| \geq Q+1$ and $d_f>d$. 
\end{theorem}

\subsection{Minimum-Distance Graphs of Perfect Codes and MDS Codes}

In this subsection we show that $G(\mathcal{C})$ is connected for the two optimal families of ECCs: 1) perfect codes and 2) Maximum distance separable (MDS) codes.

\subsubsection{Perfect Codes}
The Hamming bound on the size of an error-correcting codes $\mathcal{C}$ with length $n$ and minimum distance $d$  is given in \cite{H1986} as follows.
$$M \leq \frac{q^n}{\sum_{i=0}^{\lfloor \frac{d-1}{2} \rfloor} {n \choose i}(q-1)^i}.$$
A code that achieves the Hamming bound is called a perfect code. If $d=2t+1$ then it is called perfect $t$-error correcting code.

\begin{theorem} \label{MDG_perfect}
The minimum-distance graph of a perfect $t$-error correcting code is connected.
\end{theorem}

\begin{proof}
Let $\mathcal{C}$ be a perfect $t$-error correcting code of length $n$, and $u,v \in \mathcal{C}$. We will construct a path  in $G(\mathcal{C})$ from vertex $u$ to $v$ by moving along edges while reducing the number of coordinates where we disagree with $v$.

Let $S=\{ i\in [n] : u_i \ne v_i \}$, and clearly $|S|=d(u,v)$. If $|S|=0$, then we are already at $v$. Otherwise, $|S|\geq2t+1$ as $u$ and $v$ both are codewords.

Now repeat the following steps while $|S| \geq 2t+1$.
\begin{enumerate}
\item Choose any subset $T \subseteq S$ such that $|T|=t+1$.
\item Form a vector $x\in \mathbb{F}_q^n$ such that
$$x_i=\begin{cases}
v_i & \text{if}\ i\in T \\
u_i & \text{if}\ i\notin T
\end{cases}, \ \text{for all } \ i \in [n], $$
and we have $d(x,u)=|T|=t+1$.

\item Since $\mathcal{C}$ is a perfect code with radius $t$, the Hamming balls of radius $t$ around each codeword cover the entire space $\mathbb{F}_q^n$ without overlap. Therefore, there exists a unique $u' \in \mathcal{C}$ with $d(x,u')\leq t$.

\item By the triangle inequality, we have
$d(u,u') \leq d(x,u)+d(x,u') \leq t+1+t=2t+1.$
Since $u$ and $u'$ both are different codewords of $\mathcal{C}$, we have $d(u,u')=2t+1$. This means that the vertices $u$ and $u'$ are adjacent in $G(\mathcal{C})$.

\item Now we show that $d(u', v) < d(u,v)$.
Let $$m_{\text{in}}=|\{i \in T : u_i'\not= x_i\}|=|\{i \in T : u_i'\not=v_i\}|,$$
$$m_{\text{out}}=|\{i \in [n]\backslash T : u_i'\not= x_i\}|=|\{i \in [n]\backslash T  : u_i'\not=u_i\}|.$$
Since $d(x,u')\leq t$, we have \begin{equation}\label{m_inout}m_{\text{in}}+m_{\text{out}}\leq t.
\end{equation}
While moving from $u$ to $u'$, the distance from $v$ changes in the following ways.
\begin{itemize}
    \item \textbf{Inside $T$:} $u'$ and $v$ disagree only $m_{\text{in}}$ number of coordinates and same on others. Therefore, the total number of disagreements that are fixed now is $|T|-m_{\text{in}}$.
     \item \textbf{Outside $T$:} The maximum number of new disagreements introduced outside $T$ is $m_{\text{out}}$.
\end{itemize}
Therefore, we have 
\begin{align*}
    d(u',v)&\leq d(u,v)-(|T|-m_{\text{in}})+m_{\text{out}} \\
    & =d(u,v) -(t+1) + m_{\text{in}} +m_{\text{out}} \\
    & \leq d(u,v)-t-1+t=d(u,v)-1.
\end{align*}
\item Now the new set is $S=\{i\in [n]: u_i' \not=v_i\}$ with $|S|=d(u',v)$. 
\end{enumerate} 
Since $u'$ and $v$ are codewords, $|S|=d(u',v)\notin \{1,2,\ldots, 2t\}$. Again, if $|S|=0$, then we are at $v$, otherwise if $|S|\geq 2t+1$, repeat the same process for $u'$ instead of $u$ now.

Each move of the above process travels along an edge and strictly decreases the distance from $v$, while never entering the range $1,2,\ldots,2t$. After finitely many steps, we must reach $v$. Hence, there exists a path between $u$ and $v$, and $G(\mathcal{C})$ is connected.
\end{proof}

From Theorems \ref{MDG1} and \ref{MDG_perfect}, the following corollary holds.

\begin{corollary}\label{MDG_col1}
Let $f:\mathbb{F}_q^k \to \mathrm{Im}(f)$ be a function. Then for an $(f\!:d_d,d_f)$-FCC with $d_f>d_d$, we have
$$
r_f(k:d_d,d_f) \ge n - k + 1,
$$
where $n$ is the integer satisfying
$
q^{n-k} = \sum_{i=0}^{\lfloor \frac{d_d-1}{2}\rfloor} {n \choose i} (q - 1)^i.
$
\end{corollary}

\subsubsection{MDS Codes}

The Singleton bound on the size of an error-correcting code $\mathcal{C}$ of length $n$ and minimum distance $d$ is given in \cite{H1986} as
\[
M \le q^{\,n-d+1},
\]
equivalently, for $k=\log_q M$,
$
d \le n-k+1.
$
A code that meets the Singleton bound with equality is called a \emph{maximum distance separable (MDS)} code. In particular, an $(n,M,d)_q$ code is MDS if and only if $M=q^{\,n-d+1}$. 
The existence of MDS codes depends on the alphabet size $q$ relative to the code length, for example, Reed-Solomon and extended Reed--Solomon constructions give $q$-ary MDS codes of length up to $q+1$ \cite{LX2004}.

To prove that the minimum-distance graph of an MDS code is connected, we use the following lemma.

\begin{lemma}\label{lem:MDG_MDS}
Let $\mathcal{C}$ be an MDS code with parameters $(n,M,d)_q$, and let $u,v\in\mathcal{C}$. Then there exists $u'\in\mathcal{C}$ such that
\begin{enumerate}
    \item $d(u,u')=d$ (i.e., $u'$ is a neighbor of $u$ in $G(\mathcal{C})$);
    \item $d(u',v)\le d(u,v)-1$.
\end{enumerate}
\end{lemma}

\begin{proof}
Let $S=\{i\in[n]:u_i\ne v_i\}$ and $S^c=[n]\setminus S=\{i\in[n]:u_i=v_i\}$. Then $|S|=d(u,v)\ge d=n-k+1$, so $|S^c|=n-|S|\le k-1$.

Choose a set $J\subseteq[n]$ with $|J|=k$ such that $S^c\subseteq J$ (i.e., $J$ contains all agreements of $u$ and $v$, and adds $k-|S^c|$ positions from $S$). Then
\[
|J^c|=n-k=d-1,\]
\[|J\cap S|=k-|S^c|=k-(n-d(u,v))=d(u,v)-(d-1).
\]
Pick any $j\in J\cap S$ and define a vector $t\in\mathbb F_q^{J}$ by
$
t_j=v_j\ne u_j$ and $t_i=u_i\ \text{ for all } i\in J\setminus\{j\}.$
By the MDS projection property (for any $|J|=k$, the projection $\pi_J:\mathcal C\to \mathbb F_q^{J}$ is bijective), there exists a unique $u'\in\mathcal C$ with $\pi_J(u')=t$. Now, we have the following.
\begin{enumerate}
    \item \emph{$d(u,u')=d$.}  
The coordinates where $u$ and $u'$ may differ are: the single index $j\in J$ (forced to differ), plus the indices in $J^c$ (unconstrained). Hence
$$
d(u,u')\le 1+|J^c| = 1+(d-1)=d.
$$
Since $u'\ne u$ and $d$ is the minimum distance of $\mathcal C$, we must have $d(u,u')=d$.

\item \emph{$d(u',v)\le d(u,v)-1$.}  
We count the sure agreements between $u'$ and $v$. On $S^c\subseteq J$ we set $t_i=u_i=v_i$, so $u'_i=v_i$ for all $i\in S^c$; this gives $|S^c|=n-d(u,v)$ agreements. Additionally, at $j$ we enforced $u'_j=v_j$, yielding one more agreement. Therefore, the number of agreements is at least $n-d(u,v)+1$, and hence
\[
d(u',v)\le n-(n-d(u,v)+1) = d(u,v)-1.
\]
\end{enumerate}
This proves both claims.
\end{proof}

\begin{theorem}\label{MDG_MDS}
    The minimum-distance graph of any MDS code is connected.
\end{theorem}

\begin{proof}
Let $G(\mathcal C)$ be the minimum-distance graph of $\mathcal C$, and $u,v\in\mathcal C$. We will construct a path in $G(\mathcal{C})$ from vertex $u$ to $v$ by moving along the edges while reducing the distance from $v$. Start with $x_0=u$, and apply Lemma~\ref{lem:MDG_MDS} to the pair $(x_0,v)$ to obtain $x_1\in\mathcal C$ with
$$
d(x_0,x_1)=d \quad\text{and}\quad d(x_1,v)\le d(x_0,v)-1.
$$
If $d(x_1,v)=0$, then we are already at $v$, otherwise $d(x_1,v)\geq d$. Again use Lemma \ref{lem:MDG_MDS} for the pair $(x_1,v)$, and obtain $x_s$. By repeating this process, we get a sequence $x_0,x_1,x_2,\dots$ in $\mathcal C$ such that
\[
d(x_i,x_{i+1})=d \qquad\text{and}\qquad d(x_{i+1},v)\le d(x_i,v)-1.
\]
Because $d(x_i,v)$ is a nonnegative integer that decreases at each step, after finitely many steps we reach some $x_j$ with $d(x_j,v)=d$. 
The next application of Lemma~\ref{lem:MDG_MDS} gives us $x_{j+1}=v$ with $d(x_j,v)=d$. 
Thus $u=x_0,x_1,\dots,x_j,v$ is a path in $G(\mathcal C)$ from $u$ to $v$. Therefore, $G(\mathcal C)$ is connected.
\end{proof}
From Theorems \ref{MDG1} and \ref{MDG_MDS}, the following corollary holds.

\begin{corollary}\label{col:MDG_MDS}
Let $f:\mathbb{F}_q^k \to \mathrm{Im}(f)$ be a function. 
Assume there exists an MDS $(n,q^k,d)_q$ code, i.e., $n=k+d-1$.
Then for an $(f\!:d,d_f)$-FCC with $d_f>d$, we have
\[
r_f(k:d,d_f) \ge n-k+1 = d,
\]
equivalently, any such FCC must have length $\ge k+d$.
\end{corollary}

\section{Function-correcting codes for specific functions}\label{specific_functions}
In this section, we focus on specific classes of functions and construct FCCs for them using our method described in Subsection~\ref{subsec:Method}. The classes of functions considered are (i) Locally binary functions, (ii) Locally bounded functions, and (iii) Hamming weight functions.

\subsection{Locally binary function}

\begin{definition}[Function ball \cite{LBWY2023}]\label{FB}
The \emph{function ball} of a function $f:\mathbb{F}_q^k \rightarrow \mathrm{Im}(f)$ with radius $\rho$ around $u \in \mathbb{F}_q^k$
is defined as
$B_f (u, \rho) = \{f(u') \mid u' \in \mathbb{F}_q^k \ \text{and} \ d(u, u') \leq \rho\}.$
\end{definition}

\begin{definition}[Locally binary function \cite{LBWY2023}] \label{rl_func}
A function $f:\mathbb{F}_q^k \rightarrow \mathrm{Im}(f)$ is said to be a \emph{$\rho$-locally binary} 
function if $|B_f (u, \rho)| \leq 2$  for all $u \in \mathbb{F}_q^k.$
\end{definition}

By generalizing the optimal construction of an $(f, t_f)$-FCC for locally binary functions given in~\cite{LBWY2023}, we obtain the construction presented in the proof of the following lemma. 

\begin{lemma}\label{lem:LBF}
   For any $(d_f-1)$-locally binary function $f$, and a systematic $[n, k, d_d]$ linear error-correcting code $\mathcal{C}$, we have (from our construction in  Subsection~\ref{subsec:Method})
$$r_f (k: d_d, d_f) \le n-k+d_f-d_d.$$
\end{lemma}
\begin{proof}
Consider a $(d_f-1)$-locally binary function $f:\mathbb{F}_q^k \rightarrow \mathrm{Im}(f)$, and an $[n, k, d_d]$ linear error-correcting code $\mathcal{C}$ with a generator matrix $G$ of order $k \times n$. For any $u\in \mathbb{F}_q^k$, we denote the corresponding codeword as $c_u=uG$. Now define an encoding function $\mathfrak{C}_f: \mathbb{F}_q^k \rightarrow \mathbb{F}_q^{n+(d_f-d_d)}$ as 
$$
\mathfrak{C}_f(u)=(c_u, p_u), $$ 
where 
$$p_u = \begin{cases}
\underbrace{11\cdots 1}_{(d_f-d_d)\ \text{times}} & \text{if} \ f(u)=\max{B_f(u,d_f-1)}, \\
\underbrace{00\cdots 0}_{(d_f-d_d)\ \text{times}} & \text{otherwise}. \\
\end{cases}$$
 Now we prove that the encoding function defined above is an $(f\!:d_d,d_f)$-FCC with redundancy $r=n-k+d_f-d_d$. For any $u,v \in \mathbb{F}_q^k$,
 \begin{equation} \label{mindist}
d(\mathfrak{C}_f(u), \mathfrak{C}_f(v)) = d(c_u, c_v) + d(p_u, p_v).
\end{equation}
 Since $\mathcal{C}$ is a linear code with minimum distance $2t_d+1$, for any $u,v \in \mathbb{F}_q^k$ such that $u \neq v$, we have $d(\mathfrak{C}_f(u), \mathfrak{C}_f(v))\geq d(c_u,c_v)\geq d_d.$
 Now let $u,v \in \mathbb{F}_q^k$ be such that $f(u) \neq f(v)$.
There are the following two possible cases with vectors $u$ and $v$.\\
\textit{Case 1}: 
If $d(u,v) \geq d_f$, then $d(c_u,c_v) \geq d_f$, as $\mathcal{C}$ is a systematic code. By \eqref{mindist}, we have
$d(\mathfrak{C}_f(u), \mathfrak{C}_f(v)) = d(c_u, c_v) + d(u_p, v_p)\geq d_f.$\\
%\noindent
\textit{Case 2}: If $d(u,v)\leq d_f-1$ then $f(v) \in B_f(u, d_f-1)$, and by the definition of the function $\mathfrak{C}_f(u)$, we have $d(p_u,p_v)=d_f-d_d$, and  
$d(\mathfrak{C}_f(u), \mathfrak{C}_f(v)) = d(c_u, c_v) + d(p_u, p_v)\geq d_d+d_f-d_d=d_f.$
\end{proof}

Lemma \ref{lem:LBF} can also be written in terms of $t_d$ and $t_f$ by putting $d_d=2t_d+1$ and $d_f=2t_f+1$ as follows.
\begin{lemma*}
   For any $2t_f$-locally binary function $f$, and a systematic $[n, k, 2t_d+1]$ linear error-correcting code $\mathcal{C}$, our construction gives 
$$r_f (k, t_d, t_f) \le n-k+2(t_f-t_d).$$
\end{lemma*}

For $t_d=0$, we have $n=k$ and $r_f (k, t_d, t_f) \le 2t_f$, which is optimal from Theorem \ref{thm1.1}.

In Lemma \ref{lem:LBF}, we can use any class of systematic linear codes. If we use a linear perfect code in the proposed construction described in Subsection~\ref{subsec:Method}, then we get the following result.
\begin{corollary}\label{col:HLB}
If there exists a perfect linear $(n,q^k,d_d=2t_d+1)$-code, where 
$\sum_{i=0}^{t_d} {n \choose i}(q-1)^i = q^{n-k},$ then for any $(d_f-1)$-locally binary function $f$, 
$$r_f(k:d_d,d_f) \leq n-k+d_f-d_d.$$
\end{corollary}
Now, from Corollaries \ref{MDG_col1} and \ref{col:HLB}, we have the following result.
\begin{corollary}\label{col:opti1}
   Let $f:\mathbb{F}_q^k \to \mathrm{Im}(f)$ be a $(d_f-1)$-locally binary function. Then the construction described in the proof of Lemma \ref{lem:LBF} gives an optimal $(f\!:d_d,d_f)$-FCC for $d_f=d_d+1$ if there exists a perfect $(n, q^k, d_d)$ code.
\end{corollary}
\begin{proof}
   For a perfect $(n, q^k, d_d)$ code, we have 
   $\sum_{i=0}^{\lfloor \frac{d_d-1}{2} \rfloor} {n \choose i}(q-1)^i  = q^{n-k}.$
   Using Corollary \ref{MDG_col1}, we have $$r_f(k:d_d,d_f)\geq n-k+1.$$
   Using the proposed construction given in Subsection~\ref{subsec:Method}, we get the bound given in Corollary \ref{col:HLB} as
   $$r_f(k:d_d,d_f) \leq n-k+1.$$
   Therefore, optimality holds.
\end{proof}

\begin{example}
  Let $f:\mathbb{F}_2^4 \rightarrow \{0,1\}$ be defined by $f(u) = \mathrm{wt}(u) \bmod 2$ for all $u \in \mathbb{F}_2^4$. Clearly, this function is a locally binary function. Using our construction method, we construct an $(f, 1, 2)$-FCC. 
We use the $[7,4,3]$ Hamming code in Step 1. Then, in Step 2, we apply the method described in the proof of Lemma~\ref{lem:LBF}, which results in a redundancy of $7-4 + 2(t_f - 1) = 5$, as stated in Corollary~\ref{col:HLB}. 
The generator matrix for $[7,4,3]$ Hamming code is 
$$G=\begin{bmatrix}
    1 & 0 & 0 & 0 & 1 & 1 & 0 \\
    0 & 1 & 0 & 0 & 1 & 0 & 1 \\
    0 & 0 & 1 & 0 & 0 & 1 & 1 \\
    0 & 0 & 0 & 1 & 1 & 1 & 1 
\end{bmatrix}.$$
Further, the encoding $\mathfrak{C}_f : \mathbb{F}_2^4 \rightarrow \mathbb{F}_2^9$ is defined as $\mathfrak{C}_f(u) = (c_u, p_u)$, where $p_u = 00$ if $f(u) = 0$ and $p_u = 11$ if $f(u) = 1$. Using this encoding, we obtain the following code with redundancy of $5$.
\[
\begin{matrix}
\text{Message}  & \text{ECC-parity} & \text{FCC-parity} & \text{Codeword} \\
u \in \mathbb{F}_2^4 & c_u|_{\{5,6,7\}} & p_u & \mathfrak{C}_f(u) \\
\hline 
0000 & 000   & 00 & 000000000 \\
1000 & 110   & 11 & 100011011 \\
0100 & 101   & 11 & 010010111 \\
0010 & 011   & 11 & 001001111 \\
0001 & 111   & 11 & 000111111 \\
1100 & 011   & 00 & 110001100 \\
1010 & 101   & 00 & 101010100 \\
1001 & 001   & 00 & 100100100 \\
0110 & 110   & 00 & 011011000 \\
0101 & 010   & 00 & 010101000 \\
0011 & 100   & 00 & 001110000 \\
1110 & 000   & 11 & 111000011 \\
1101 & 100   & 11 & 110110011 \\
1011 & 010   & 11 & 101101011 \\
0111 & 001   & 11 & 011100111 \\
1111 & 111   & 00 & 111111100 
\end{matrix}
\]

where $c_u|_{\{5,6,7\}}$ denotes the restriction of the vector $c_u=uG$ to its coordinates at positions  5, 6 and 7.

By Theorem \ref{thm1.2}, we have $r_f(4,1,2) \geq 5$, which establishes that the constructed code is optimal.
\end{example}

If we use an MDS code in the construction given in Lemma \ref{lem:LBF},  we get the following result.
\begin{corollary}\label{col:HLB_MDS}
If there exists an MDS $(n,q^k,d_d=n-k+1)$-code, then for any $(d_f-1)$-locally binary function $f$, 
$$r_f(k:d_d,d_f) \leq n-k+d_f-d_d=d_f-1.$$
\end{corollary}
Now, from Corollaries \ref{col:MDG_MDS} and \ref{col:HLB_MDS}, we have the following result.
\begin{corollary}\label{col:opti2}
   Let $f:\mathbb{F}_q^k \to \mathrm{Im}(f)$ be a $(d_f-1)$-locally binary function. Then the construction described in the proof of Lemma~\ref{lem:LBF} gives an optimal $(f\!:d_d,d_f)$-FCC for $d_f=d_d+1$, if there exists an $(n, q^k, d_d)$ MDS code.
\end{corollary}
\begin{proof}
   For an $(n, q^k, d_d)$ MDS code, we have 
   $d_d=n-k+1$.
   Using Corollary \ref{col:MDG_MDS}, we have $$r_f(k:d_d,d_f)\geq d_d.$$
   Using the proposed construction, we get the bound given in Corollary \ref{col:HLB_MDS} as
   $$r_f(k:d_d,d_f) \leq d_f-1=d_d.$$
   Therefore, optimality holds.
\end{proof}

%%%%%%%%%%%%%%%%%%%%%%%%%%%%%%%%%%%%%%%%%%%%%%%%%%%%%

\subsection{Locally bounded functions}
Locally bounded functions are a generalization of the locally binary functions introduced in~\cite{RRHH2025}, and are defined as follows.

\begin{definition}[Locally bounded function] \label{LBo_func}
A function $f:\mathbb{F}_q^k \rightarrow \mathrm{Im}(f)$ is said to be a \emph{$(\rho, \lambda)$-bounded function} if $|B_f (u, \rho)| \leq \lambda$ for all $u \in \mathbb{F}_q^k.$
\end{definition}

 The following two lemmas are used in the proof of Theorem~\ref{lem:LBoF}.

\begin{lemma}\label{lem:Col}\cite{RRHH2025}
Let $f:\mathbb{F}_2^k \to S$ be a locally $(\rho,\lambda)$-bounded function.
Assume that there exists a total order $\prec$ on $\mathrm{Im}(f)$ such that for every $u \in \mathbb{F}_2^k$, the set $B_f(u,\rho)$ forms a contiguous block with respect to $\prec$. Then there exists a mapping
$
\mathrm{Col}_f : \mathbb{F}_2^k \to [\lambda]
$
such that for any $u,v \in \mathbb{F}_2^k$ with $d(u,v) \le \rho$ and $f(u)\neq f(v)$, we have $\mathrm{Col}_f(u)\neq \mathrm{Col}_f(v)$.
\end{lemma}
 
\begin{lemma}\label{N_4}\cite{RRHH2025}
Let $N(\lambda, 2t)$ be the minimum length of a binary error-correcting code with $\lambda$ codewords and minimum distance $2t$. Then $N(4, 2t)=3t$.
\end{lemma}

The construction for FCC given for locally binary functions in the proof of Lemma \ref{lem:LBF} can be generalized for $(2t_f, \lambda)$-bounded functions, and we get the following result.

\begin{theorem}\label{lem:LBoF}
   For any \((2t_f, \lambda)\)-bounded function \(f\) satisfying the contiguous block condition given in Lemma~\ref{lem:Col}, and a systematic $[n, k, 2t_d+1]$ linear error-correcting code $\mathcal{C}$, we have
$$r_f (k, t_d, t_f) \le n-k+N(\lambda, 2(t_f-t_d)),$$
where $N(\lambda, d)$ is the minimum length of an error-correcting code with $\lambda$ codewords and minimum distance $d$.
Particularly, for $q=2$ and $\lambda=4$,
$$r_f (k, t_d, t_f) \le n-k+3(t_f-t_d).$$
\end{theorem}

\begin{proof}
The proof follows in a similar manner as the proof of Lemma \ref{lem:LBF}. 
Let $f$ be a locally $(2t_f, \lambda)$-bounded function, and $\mathcal{C}$ be an $[n, k, 2t_d+1]$ linear systematic error-correcting code with a generator matrix $G$ of order $k \times n$. For any $u\in \mathbb{F}_q^k$, we denote the corresponding codeword as $c_u=uG$. 

By Lemma \ref{lem:Col}, there exists a mapping $\mathrm{Col}_f:\mathbb{F}_q^k \rightarrow [\lambda]$ for  a function $f$ such that for any $u,v \in \mathbb{F}_q^k $,  $\mathrm{Col}_f (u) \neq \mathrm{Col}_f(v)$ if $f(u)\ne f(v)$ and $d(u,v)\leq 2t_f$. Let $\mathcal{C}'$ be an error-correcting code with $\lambda$ codewords, minimum distance ${d'}=2(t_f-t_d)$ and length $N(\lambda, d')$, and let the codewords of $\mathcal{C}'$ be denoted by $c'_1, c'_2, \ldots, c'_{\lambda}$.
Define an encoding function $\mathfrak{C}_f: \mathbb{F}_q^k \rightarrow  \mathbb{F}_q^{n+N(\lambda, d')}$ as 
$$\mathfrak{C}_f(u)=(c_u, p_u), \ \text{where}  \ p_u =c'_{ \mathrm{Col}_f(u)}.$$
 Now we prove that the encoding function defined above is an $(f, t_d, t_f)$-FCC with redundancy $r=n-k+N(\lambda, d')$. 
 For any $u,v \in \mathbb{F}_q^k$,
 $
d(\mathfrak{C}_f(u), \mathfrak{C}_f(v)) = d(c_u, c_v) + d(p_u, p_v).
$
 Since $\mathcal{C}$ is a linear code with minimum distance $2t_d+1$, for any $u,v \in \mathbb{F}_q^k$ such that $u \neq v$, we have $d(\mathfrak{C}_f(u), \mathfrak{C}_f(v))\geq d(c_u,c_v)\geq 2t_d+1.$
 Now let $u,v \in \mathbb{F}_q^k$ be such that $f(u) \neq f(v)$.
Then we have the following two possible cases with vectors $u$ and $v$.

\textit{Case 1}: If $d(u,v) \geq 2t_f+1$, then we have
$d(\mathfrak{C}_f(u), \mathfrak{C}_f(v)) = d(c_u, c_v) + d(p_u, p_v)\geq 2t_f+1,$ as the code $\mathcal{C}$ is in systematic form.

\textit{Case 2}: If $d(u,v) \leq 2t_f$ then $f(v) \in B_f(u, 2t_f)$, and by the definition of  $\mathrm{Col}_f$, we have $\mathrm{Col}_f(u) \neq \mathrm{Col}_f(v)$. Therefore, $d(p_u, p_v) =d(c'_{\mathrm{Col}_f(u)}, c'_{\mathrm{Col}_f(v)}) \geq d'=2(t_f-t_d)$ as the minimum distance of the code $\mathcal{C}'$ is $d'$. Since $u \neq v$, we have $d(c_u,c_v) \geq 2t_d+1$ and 
$d(\mathfrak{C}_f(u), \mathfrak{C}_f(v)) = d(c_u, c_v) + d(p_u, p_v)\geq 2t_d+1+2(t_f-t_d)=2t_f+1.$

Further, from Lemma \ref{N_4}, we have $N(4, 2(t_f-t_d))=3(t_f-t_d)$ for $q=2$. Therefore, for $\lambda=4, q=2$, it follows that $r_f (k, t_d, t_f) \le n-k+3(t_f-t_d).$
\end{proof}

\subsection{Hamming weight function}
Hamming weight function is defined as $f:\mathbb{F}_q^k\rightarrow \{0,1,2,\ldots,k\},\ f(u)=wt(u)$, where $wt(u)$ denotes the Hamming weight of the vector $u$. As clear from the definition, this function is a $(2t_f, 4t_f+1)$- bounded function, and for
the Hamming weight function each neighborhood of $B_f (u, 2t_f)$ forms a contiguous block in the natural order of weights. Therefore, using Theorem  \ref{lem:LBoF}, we get 
$$r_f (k, t_d, t_f) \le n-k+N(4t_f+1, 2(t_f-t_d)).$$
However, we can find a tighter bound on the redundancy of the FCC for Hamming weight function given as follows.
\begin{lemma}\label{lem:HF}
   For Hamming weight function $f:\mathbb{F}_q^k \rightarrow \mathrm{Im}(f)$, and a systematic $[n, k, 2t_d+1]$ linear error-correcting code $\mathcal{C}$, we have
$$r_f (k, t_d, t_f) \le n-k+N(2t_f+1, 2(t_f-t_d)),$$
or 
$$r_f (k, t_d, t_f) \le N(q^k, 2t_d+1)+N(2t_f+1, 2(t_f-t_d))-k,$$
where $N(\lambda, d)$ is the minimum length of an error-correcting code with $\lambda$ codewords and minimum distance $d$.
\end{lemma}
\begin{proof}
      Let $G$ be a generator matrix of code $\mathcal{C}$, and for any $u\in \mathbb{F}_q^k$, $c_u=uG$. Furthermore, let $\mathcal{C}'$ be an error-correcting code with $2t_f+1$ codewords, minimum distance $2(t_f-t_d)$, and length $N(2t_f+1, 2(t_f-t_d))$. Denote the codewords of $\mathcal{C}'$ by $c'_0, c'_1, \ldots, c'_{2t_f}$. Define an encoding function $\mathfrak{C}_f: \mathbb{F}_q^k \rightarrow  \mathbb{F}_q^{n+N(2t_f+1, 2(t_f-t_d))}$ as 
$$\mathfrak{C}_f(u)=(c_u, p_u), \ \text{where} \  p_u = c'_{f(u) \bmod{(2t_f+1)}}.$$
 Now we prove that the encoding function defined above is a $(f, t_d, t_f)$-FCC. Since $\mathcal{C}$ is a linear code with minimum distance $2t_d+1$, for any $u,v \in \mathbb{F}_q^k$ such that $u \neq v$, we have $d(\mathfrak{C}_f(u), \mathfrak{C}_f(v))\geq d(c_u,c_v)\geq 2t_d+1.$
 Now let $u,v \in \mathbb{F}_q^k$ be such that $f(u) \neq f(v)$. Then we have the following cases.

%\noindent
\textit{Case 1}: If $0<|f(u)-f(v)| \leq 2t_f$, then
$c'_{f(u) \bmod{(2t_f+1)}} \neq c'_{f(v) \bmod{(2t_f+1)}}$. Since the minimum distance of the code $\mathcal{C}'$ is $2(t_f-t_d)$, we have $d(p_u, p_v) \geq 2(t_f-t_d)$. Since $u \neq v$, we have $d(c_u,c_v) \geq 2t_d+1$ and 
$d(\mathfrak{C}_f(u), \mathfrak{C}_f(v)) = d(c_u, c_v) + d(p_u, p_v)\geq 2t_f+1.$

%\noindent
\textit{Case 2}:  If $|f(u) - f(v)| > 2t_f$, then WLOG assuming $f(u)>f(v)$, we have $f(u)-f(v)  \geq 2t_f+1$. 
Since $d(u,v) \geq wt(u)-wt(v)$, we have 
$d(u, v) \geq 2t_f+1$, and
$d(\mathfrak{C}_f(u), \mathfrak{C}_f(v)) = d(c_u, c_v) + d(p_u, p_v)\geq 2t_f+1.$
\end{proof}

\begin{example}
    Consider the Hamming weight function $f$ over the space $\mathbb{F}_2^8$, and let $t_d = 1$ and $t_f = 2$. For a linear code, we have $N(q^k, 2t_d + 1) = N(2^8, 3) = 12$. Additionally, $N(2t_f + 1, 2(t_f - t_d)) = N(5, 2) = 4$. Therefore, the redundancy bound satisfies  
$$r_f(8, 1, 2) \leq 12 + 4 - 8 = 8.$$
Furthermore, an $(f, 1, 2)$-FCC with this level of redundancy can be constructed using our proposed method: by employing a $[12, 8, 3]$ linear code in Step 1, and a $(4, 5, 2)$ code as $\mathcal{C}'$ in Step 2, following the procedure explained in the proof of Lemma~\ref{lem:HF}.
\end{example}

%\clearpage
%%%%%%%%%%%%%%%%%%%%%%%%%%%%%%%%%%%%%%%%%%%%%%%%%%%%%%%%%

\section{Linear $(f: d_d,d_f)$-FCC} \label{linearFCC}
In this section, we consider linear function-correcting codes with data protection, where linearity is defined in the same way as in traditional error-correcting codes. We then focus on linear $(f\!:d_d,d_f)$-FCCs for linear functions, and present their connection to coset codes.

\begin{definition}
    For a function $f:\mathbb{F}_q^k \rightarrow \mathrm{Im}(f)$ and two nonzero integers $d_d$ and $d_f$ such that $d_d < d_f$, an $(f\!:d_d,d_f)$-FCC of length $n$ is called a linear FCC if it is a subspace of the vector space $\mathbb{F}_q^n$ with dimension $k$.
\end{definition}
As clear from the definition, a linear $(f\!:d_d,d_f)$-FCC is also a linear $[n,k,d_d]$ error-correcting code. Any linear code can be expressed using its generator matrix $G$ of size $k \times n$ or its parity-check matrix $H$ of size $(n-k) \times n$. If the generator matrix $G$ has the form $G=[I_k \mid P]$, where $I_k$ denotes the $k \times k$ identity matrix and $P$ is some $k \times (n-k)$ matrix, then $G$ is in standard form. From \cite[Theorem 4.6.3]{LX2004}, we know that any linear code $C$ is equivalent to a linear code $C'$ with a generator matrix in standard form. So in our further discussions, we will only consider linear codes in standard form unless otherwise mentioned specifically. Therefore, a linear FCC can also be defined as follows.

   For a function $f:\mathbb{F}_q^k \rightarrow \mathrm{Im}(f)$, a subspace $\mathcal{C}$ of the vector space $\mathbb{F}_q^n$ with dimension $k$ is called a linear $(f\!:d_d,d_f)$-FCC if 
  \begin{enumerate}
      \item weight of any non-zero codeword in $\mathcal{C}$ is at least $d_d$.
      \item For any $c_1=(u_1,p_1), c_2=(u_2,p_2) \in \mathcal{C}$ such that $f(u_1)\neq f(u_2)$, we have $d(c_1,c_2)\geq d_f.$
  \end{enumerate}

We note that linear $(f\!:d_d,d_f)$-FCCs can exist for both linear and nonlinear functions $f$. The following example illustrates that even when $f$ is nonlinear, it is possible to obtain a linear FCC. In the remainder of the paper, however, we will focus on linear FCCs corresponding to linear functions.
\begin{example}
Consider the nonlinear function $f:\mathbb{F}_2^3\to\mathbb{F}_2$ defined as $f(x_1,x_2,x_3)=x_1x_2.$
 Let $\mathcal C$ be the binary linear code generated by
$$
G=\begin{bmatrix}
0&0&1&1&1&1&0&0&1\\
1&1&0&0&1&1&0&1&0\\
0&0&0&0&0&0&1&1&1
\end{bmatrix}.
$$
The minimum distance of $\mathcal{C}$ is $d_d=3$. Further, whenever $f(u)\neq f(v)$, the corresponding codewords satisfy $d(c_u,c_v) \geq 5.$
Therefore, $d_f=5$, and $\mathcal{C}$ is a linear $(f:3,5)$-FCC. The corresponding function values and codewords for the message vectors are given below.
$$
\begin{matrix}
\text{Message $u$} & \text{Function value }f(u) & \text{Codeword }c_u \\
000 & 0 & 000000000\\
001 & 0 & 000000111\\
010 & 0 & 110011010\\
011 & 0 & 110011101\\
100 & 0 & 001111001\\
101 & 0 & 001111110\\
110 & 1 & 111100011\\
111 & 1 & 111100100
\end{matrix}
$$
\end{example}

\subsection{Linear $(f\!:d_d,d_f)$-FCC for linear functions}

If $f:\mathbb{F}_q^k \rightarrow \mathrm{Im}(f)$ is a linear function, then Ker$(f)$ is a subspace of $\mathbb{F}_q^k$ and the coset partition of Ker$(f)$ partitions the space $\mathbb{F}_q^k$.

\begin{lemma}
    If $f:\mathbb{F}_q^k \rightarrow \mathrm{Im}(f)$ is a linear function and $\mathcal{C}$ is a linear $(f\!:d_d,d_f)$-FCC of length $n$ in standard form, then the subcode $D_f=\{c=(u,p) \in \mathcal{C} \mid u \in \text{Ker}(f)\}$ forms a subspace of $\mathcal{C}$. Furthermore, the dimension of $D_f$ is the same as the dimension of Ker$(f)$.
\end{lemma}

\begin{proof}
    Let $d_1,d_2 \in D_f$, then $d_1=(u_1, p_1)$ and $d_2=(u_2,p_2)$ for some $u_1, u_2 \in \text{Ker}(f)$. For $\lambda, \mu \in \mathbb{F}_q$, we have 
    $$\lambda d_1+ \mu d_2=\lambda (u_1, p_1) + \mu(u_2, p_2)= (\lambda u_1+ \mu u_2, \lambda p_1+ \mu p_2).$$
    Since Ker$(f)$ is a subspace of $\mathbb{F}_q^k$, we have $\lambda u_1+ \mu u_2 \in \text{Ker}(f)$, and hence $\lambda d_1+ \mu d_2 \in D_f$. Therefore, $D_f$ is a subspace of $\mathcal{C}$.
\end{proof}

From the above lemma, a linear $(f;d_d, d_f)$-FCC for a linear function $f$ has a minimum distance $d_d$ and the minimum distance between cosets of $D_f$ is $d_f$, that is, for any $c_1 \in v_i+D_f$ and $c_2 \in v_j+D_f$, we have $d(c_1,c_2) \geq d_f$, where $v_i+D_f$ and $v_j+D_f$ are two different cosets of $D_f$ in $\mathcal{C}$. Since $d(x,y)=wt(x-y)$ for any two vectors $x$ and $y$, we have
$$\min \{\text{wt}(c_1 - c_2) \mid c_1 \in v_i+D_f,c_2\in v_j+D_f\} \geq d_f.$$

\begin{lemma}
  % Let $\mathcal{C}$ be a linear $(f, t_d, t_f)$-FCC. For any $v_i, v_j \in \mathcal{C}$ such that $v_i\ne v_j$,
  Let $\mathcal{C}$  be a subspace of $\mathbb{F}_q^n$ and $D$ be a subspace of $C$. Then for any $v_i, v_j \in \mathcal{C}$ such that $v_i\ne v_j$,
   $$
      \min \{\text{wt}(c_1 - c_2) \mid c_1 \in v_i+D,c_2\in v_j+D\} = \min_{c\in (v_i-v_j)+D}\text{wt}(c)\\= \min_{d\in D}\text{wt}(v_i-v_j+d). 
   $$
\end{lemma}
\begin{proof}
    This is easy to prove as $(x+D)+(y+D)=(x+y)+D$ for any $x,y\in C$.
\end{proof}

This means the second condition in a linear $(f; d_d, d_f)$-FCC for a linear function $f$ will be 
$$\min_{d\in D_f}\text{wt}(v_i-v_j+d) \geq d_f,$$
where $v_i$ and $v_j$ are the vectors in different cosets of $D_f$, i.e., $v_i-v_j \notin D_f$.

Since $D_f$ is a subspace, it contains zero vector. Therefore, to maintain the coset distance requirement, the following result holds.

\begin{lemma}
   Let $\mathcal{C}$ be a linear $(f\!:d_d,d_f)$-FCC of a linear function $f$. Then $\text{wt}(v) \geq d_f$ for all $v\in \mathcal{C}\setminus D_f$.
\end{lemma}

\begin{example}
Consider a linear function $f:\mathbb{F}_2^2 \rightarrow \mathrm{Im}(f)$ with Ker$(f)=\{00, 11\}$. Then the following linear code $\mathcal{C}$ has the minimum distance $d_d=2$ and the minimum coset distance $d_f=3$.
$$\mathcal{C}=\{0000, 0111, 1100, 1011\}.$$
Therefore, it is an $(f:2,3)$-FCC.
\end{example}

Since the second condition in a linear 
 $(f: d_d,d_f)$-FCC for a linear function $f$ is entirely determined by the coset structure of Ker$(f)$, we can equivalently redefine these codes in terms of coset codes as follows.

\begin{definition}[Coset code]
    Let $ C \subseteq \mathbb{F}_q^n $ be a subspace of the vector space $ \mathbb{F}_q^n $, and $ D \subseteq C$ be a subspace of $C$.
The coset code $C/D$ is the collection of all cosets of $D$ in $C$, i.e.,
$$C/D=\{v+D \mid v \in C\}$$
The set of all such cosets partitions the entire subspace $C$ into $q^{k-r}$ disjoint subsets, where Dim$(C)=k$ and Dim$(D)=r$. 
\end{definition}

Here, we follow \cite[ch. 3]{MakkonenPhD2025} that provides a convenient collection of definitions and results related to coset codes matching our purpose.
%Several results on coset codes are developed in Chapter~3 of~\cite{MakkonenPhD2025}. 
The minimum distance of a coset code is defined as
$$d(C/D)=\min_{u,v\in C, u-v \notin D}\{\min\{d(c_1,c_2) \mid c_1 \in u+D, c_2 \in v+D\}\}.$$

\begin{definition}[Linear $(f\!:d_d,d_f)$-FCC for linear function]\label{def:linearFCC}
     For a linear function $f:\mathbb{F}_q^k \rightarrow \mathrm{Im}(f)$, a subspace $\mathcal{C}$ of the vector space $\mathbb{F}_q^n$ with dimension $k$ is called a linear $(f\!:d_d,d_f)$-FCC if 
  \begin{enumerate}
      \item $d(\mathcal{C})\geq d_d$.
      \item $d(\mathcal{C}/D_f)\geq d_f,$ where $D_f=\{c=(u,p) \in \mathcal{C} \mid u \in \text{Ker}(f)\}$
  \end{enumerate}
\end{definition}

Using the following straightforward lemma, we demonstrate how linear FCCs for linear functions can be systematically obtained via our proposed two-step construction method.

\begin{lemma}[Linearity of Concatenated Code via a Linear Function]\label{Linear_cat}
Let 
$f:\mathbb{F}_q^k \to \mathbb{F}_q^{\ell}$ be a linear function. 
Further, let $\mathcal{C} \subseteq \mathbb{F}_q^n$ and 
$\mathcal{D} \subseteq \mathbb{F}_q^{r'}$ be linear codes of dimensions 
$k$ and $\ell$, respectively, represented by the linear maps
$
C:\mathbb{F}_q^k \to \mathbb{F}_q^{n},\ 
D:\mathbb{F}_q^{\ell} \to \mathbb{F}_q^{r'}.
$
as usual. Then the concatenated code
\[
\mathcal{C}_{\mathrm{cat}} 
= \{\, (C(u),\, D(f(u))) \;:\; u \in \mathbb{F}_q^k \,\}
\]
is a linear code in $\mathbb{F}_q^{\,n+r'}$ of dimension $k$.
\end{lemma}
\begin{comment}
\begin{proof}
Since $C(0)=0, D(0)=0$ and $f(0)=0$, we have $(C(0),D(f(0)))=(0,0)\in\mathcal{C}_{\mathrm{cat}}$, so the set is nonempty.
Let $x=(C(u),D(f(u)))$ and $y=(C(v),D(f(v)))$ be arbitrary elements of $\mathcal{C}_{\mathrm{cat}}$, and let
$\alpha,\beta\in\mathbb{F}_q$. Using linearity of $C$, $f$, and $D$,
\[
\alpha x+\beta y
= \big(\alpha C(u)+\beta C(v),\ \alpha D(f(u))+\beta D(f(v))\big)
= \big(C(\alpha u+\beta v),\ D(f(\alpha u+\beta v))\big)
\in \mathcal{C}_{\mathrm{cat}}.
\]
Thus $\mathcal{C}_{\mathrm{cat}}$ is closed under all linear combinations and contains the zero vector; hence it is a subspace of $\mathbb{F}_q^{\,n+r'}$, i.e., a linear code.
For any $u,v\in \mathbb{F}_q^k$, if $(C(u), D(f(u)))=(C(v), D(f(v)))$ then $C(u)=C(v)$ which implies that $u=v$ as $\mathcal{C}$ is a linear code of dimension $k$. Hence $\dim(\mathcal{C}_{\mathrm{cat}})=k.$
\end{proof}
\end{comment}
%\begin{remark}
Intuitively, the linear function $f$ partitions $C$ into $q^{\ell}$ subsets, and 
the appended block $D(f(u))$ depends \emph{linearly} on the subset label.
The mapping $u \mapsto (C(u), D(f(u)))$ is therefore linear, so its image $C_{\mathrm{cat}}$ is again a linear code.
If either $C$, $D$ or $f$ were nonlinear, this linearity would fail, and $C_{\mathrm{cat}}$ would not in general be a subspace.
%\end{remark}

%%%%%%%%%%%%%%%%%%%%%
\subsection{Construction of Linear FCCs for Linear Functions}\label{Linear_FCC_construction}
\noindent\textbf{Construction for linear $(f:\!d_d,d_f)$-FCC for linear function:}
Let $f:\mathbb{F}_q^k \to \mathbb{F}_q^{\ell}$ be a linear function. Let $\mathcal{C}\subseteq\mathbb{F}_q^n$ be a systematic $[n,k,d_d]$ linear code, and let 
$\mathcal{D}\subseteq\mathbb{F}_q^{r'}$ be a linear $[r',\ell,\,d_f-d_d]$ code.
Define the encoding
\[
\mathfrak{C}_f:\mathbb{F}_q^k\to\mathbb{F}_q^{\,n+r'},\qquad
u\mapsto \big(C(u),\,D(f(u))\big),
\]
where $C$ and $D$ are linear encoders for $\mathcal{C}$ and $\mathcal{D}$, respectively. 
The image 
$
\mathcal{C}_{\mathrm{cat}}=\{ \mathfrak{C}_f(u) : u\in\mathbb{F}_q^k \}
$
is an $(f\!:\!d_d,d_f)$-FCC. 
%As given in Definition~\ref{def:linearFCC}, for this code we have
%$
%D_f=\{  \mathfrak{C}_f(u) : u \in \mathrm{Ker}(f) \},
%$
%and the minimum distances of $\mathcal{C}_{\mathrm{cat}}$ and coset code $\mathcal{C}_{\mathrm{cat}}/D_f$ satisfy 
%$
%d(\mathcal{C}_{\mathrm{cat}})\ge d_d,\ 
%d(\mathcal{C}_{\mathrm{cat}}/D_f)\ge d_f.
%$

\begin{theorem}[Correctness]
The image $\mathcal{C}_{\mathrm{cat}}=\{\,(C(u),D(f(u))) : u\in\mathbb{F}_q^k\,\}\subseteq\mathbb{F}_q^{\,n+r'}$ is a linear
$(f\!:\!d_d,d_f)$-FCC of dimension $k$ and total redundancy $r_s=(n-k)+r'$.
\end{theorem}

\begin{proof}
The linearity and dimension of the code $\mathcal{C}_{\mathrm{cat}}$ follow directly from Lemma~\ref{Linear_cat}. We now prove the distance guarantees of this code.
Let $x_1=(C(u_1),D(f(u_1))), x_2=(C(u_2), D(f(u_2)))\in\mathcal{C}_{\mathrm{cat}},$ for some $u_1,u_2\in\mathbb{F}_q^k$.
Then the Hamming distance between them is
\[
d(x_1,x_2)=d\big(C(u_1),C(u_2)\big)+d\big(D(f(u_1)),D(f(u_2))\big).
\]
If $f(u_1)=f(u_2)$, then the second term is $0$, while 
$d\big(C(u_1),C(u_2)\big)\ge d_d$, hence $d(x_1,x_2)\ge d_d$.
If $f(u_1)\ne f(u_2)$, then
$$
d\big(C(u_1),C(u_2)\big)\ge d_d $$
and
$$d\big(D(f(u_1)),D(f(u_2))\big)\ge d(\mathcal{D})=d_f-d_d,
$$
so $d(x_1,x_2)\ge d_d+(d_f-d_d)=d_f$.
Thus, $\mathcal{C}_{\mathrm{cat}}$ is an $(f\!:\!d_d,d_f)$-FCC.
\end{proof}

\begin{example}\label{ex:linFCC_2x3}
Let $f:\mathbb{F}_2^3\to\mathbb{F}_2^2$ be a linear function defined as
\[
f(u)=\begin{bmatrix}1&1&0\\[2pt]0&0&1\end{bmatrix}u,
\]
i.e., $f(u_1,u_2,u_3)=(u_1\oplus u_2,\ u_3)$.

Using the construction presented in Subsection \ref{Linear_FCC_construction}, we construct an $(f:4,6)$-FCC using the following codes.

\begin{itemize}
\item \textbf{Step 1 (Code $\mathcal{C}$):} 
A systematic binary $[7,3,4]$ code $\mathcal{C}$ with generator
\[
G_{\mathcal{C}}=\begin{bmatrix}
1&0&0&1&1&0&1\\
0&1&0&1&0&1&1\\
0&0&1&0&1&1&1
\end{bmatrix}.
\]

\item \textbf{Step 2 (Code $\mathcal{D}$):}
Choose a binary $[3,2,2]$ code $\mathcal{D}$ with generator matrix
\[
G_{\mathcal{D}}=\begin{bmatrix}
1&0&1\\
0&1&1
\end{bmatrix}.
\]
\end{itemize}

The final linear $(f:4,6)$-FCC, which is also a $[10,3,4]$ linear code, is given as:

\[ \begin{matrix} \text{Message} & \text{$\mathcal{C}$-parity} & \text{$\mathcal{D}$-parity} & \text{Codeword} \\
000 & 0000 & 000 & 0000000\,000\\
110 & 0110 & 000 & 1100110\,000\\
100 & 1101 & 101 & 1001101\,101\\
010 & 1011 & 101 & 0101011\,101\\
001 & 0111 & 011 & 0010111\,011\\
111 & 0001 & 011 & 1110001\,011 \\
101 & 1010 & 110 & 1011010\,110\\
011 & 1100 & 110 & 0111100\,110
\end{matrix}
\]

There exists a $[10,3,5]$ code \cite{CodeTable}. However, to construct a code with a minimum distance of $6$, one requires a length of at least $11$. Hence, it is evident that achieving the desired pair of values $(d_d, d_f)$ involves an inherent trade-off between data and function protection.
\end{example}

%%%%%%%%%%%%%%%%%%%%%%%%%%%%%%%%%%%%%%%%%%%%%%%%%%%%%%%%

%\clearpage

\section{Extension of Bounds from Error-Correcting Codes to Function-Correcting Codes}\label{bounds}

In this section, we consider certain bounds on error-correcting codes and extend them to the cases of $(f, t_f)$-FCCs and $(f, t_d, t_f)$-FCCs.

The generalized Plotkin bound and Gilbert-Varshamov bound for binary FCC given in \cite{LBWY2023} will work for the new distance requirement matrix $D_f(t_d, t_f: u_1,u_2,\ldots, u_{2^k})$
for an $(f,t_d,t_f)$-FCC as follows.

\begin{lemma}[\!\cite{LBWY2023}] \label{Lem:Plot_from_[1]}
For any distance matrix $\mathcal{D} \in \mathbb{N}^{M \times M}$,
$$N(\mathcal{D}) \geq \begin{cases}
\frac{4}{M^2} \sum_{i,j,i<j} [D]_{i, j} & \text{if $M$ even,} \\
\frac{4}{M^2-1} \sum_{i, j, i<j} [D]_{i, j} & \text{if $M$ odd}.
\end{cases}$$
\end{lemma}

\begin{lemma}[\!\cite{LBWY2023}] 
For any distance matrix $D \in \mathbb{N}_0^{M\times M}$, and any
permutation $\pi : [M] \rightarrow [M]$
\begin{equation*}
    N(\boldsymbol{D}) \leq \min _{r \in \mathbb{N}}\left\{r: 2^{r}>\max _{j \in[M]} \sum_{i=1}^{j-1} V\left(r,[\boldsymbol{D}]_{\pi(i) \pi(j)}-1\right)\right\} .
\end{equation*}
    
\end{lemma}
%%%%%%%%%%%%%%%%%%%%%%%%%%%%%%%%%

\subsection{Generalized Plotkin bound for $(f, d_d, d_f)$-FCCs }
%The generalized Plotkin bound for binary FCC given in \cite{LBWY2023} will work for the new distance matrix $D_f(t_d, t_f: u_1,u_2,\ldots, u_{2^k})$ for an $(f,t_d,t_f)$-FCC over $\mathbb{F}_2$ as follows.

%%%%%%%%%%%%%%%%%%%%%%%%%%%%
In the following lemma, we extend the result of Lemma \ref{Lem:Plot_from_[1]}  to the arbitrary field $\mathbb{F}_q.$

%%%%%%%%%%%%%%%%%%%?????????????????%%%%%%%%%%%%%%%%%%%%%%%%%
\begin{lemma}\label{lem:plot_gen}
	For any distance matrix $D\in \mathbb{N}^{M\times M}$ and for irregular
	distance codes over $\mathbb{F}_q$, we have
	\[
	N(D) \ge
	\frac{2q}{M^2(q-1)-a(q-a)}
	\sum_{1 \le i < j \le M} [D]_{i,j},
	\]
	where $a = M \bmod q$.
\end{lemma}

\begin{proof}
	Let $p_1,p_2,\ldots,p_M$ be a $D$-code of length $r$.  
	Consider the $j$-th coordinate of all vectors. %denoted  $(p_1^j,p_2^j,\ldots,p_M^j)$.  
	The contribution of this $j$-th coordinate to the total pairwise distance 
	$
	\sum_{1 \le i < j \le M} d(p_i,p_j),
	$ is maximized when the symbols of $\mathbb{F}_q$ appear as evenly as possible.
	Let $a = M \bmod q$, then in the most evenly distributed case, $a$ symbols appear $\left\lceil \frac{M}{q}\right\rceil$ times and
	the remaining $(q-a)$ symbols appear $\left\lfloor \frac{M}{q}\right\rfloor$ times.
	Thus,
	$$
	\sum_{1 \le i < j \le M} d(p_i,p_j)
	\le
	\frac{r}{2}\left[
	a\Big\lceil \tfrac{M}{q}\Big\rceil\Big(M-\Big\lceil \tfrac{M}{q}\Big\rceil\Big)
	+ \right.\\
	\left. (q-a)\Big\lfloor \tfrac{M}{q}\Big\rfloor\Big(M-\Big\lfloor \tfrac{M}{q}\Big\rfloor\Big)
	\right].
	$$
	On the other hand, by the definition of a $D$-code,
	\[
	d(p_i,p_j)\ge [D]_{i,j} \quad \text{for all} \quad 1\le i,j\le M.
	\]
	Hence,
	\[
	\sum_{1 \le i < j \le M} [D]_{i,j}
	\le
	\frac{r}{2}\,A,
	\]
	where
	$
	A
	=
	a\Big\lceil \tfrac{M}{q}\Big\rceil\Big(M-\Big\lceil \tfrac{M}{q}\Big\rceil\Big)
	+
	(q-a)\Big\lfloor \tfrac{M}{q}\Big\rfloor\Big(M-\Big\lfloor \tfrac{M}{q}\Big\rfloor\Big).
	$
	
	Using 
	$
	\Big\lfloor \tfrac{M}{q}\Big\rfloor=\frac{M-a}{q}=\alpha,
	\Big\lceil \tfrac{M}{q}\Big\rceil=\frac{M+q-a}{q}=\alpha+1,
	$
	we get
	\begin{align*}
		A &=  a(\alpha +1)(M-(\alpha+1)) + (q-a) \alpha (M-\alpha)\\
		&=Ma+Mq\alpha-2a\alpha -q\alpha^2-a \\
		&= Ma + Mq\left( \frac{M-a}{q} \right)-2a \left( \frac{M-a}{q} \right)-q\left( \frac{M-a}{q} \right)^2-a \\
		&=\frac{1}{q}\left(M^2(q-1)-a(q-a)\right).
	\end{align*}
	
	Therefore,
	\[
	\sum_{1 \le i < j \le M} [D]_{i,j}
	\le
	\frac{r}{2}\cdot \frac{1}{q}\left(M^2(q-1)-a(q-a)\right).
	\]
	
	Since this is true for any $D$-code of length $r$, we have the desired bound as
	\[
	N(D)
	\ge
	\frac{2q}{M^2(q-1)-a(q-a)}
	\sum_{1 \le i < j \le M} [D]_{i,j}.
	\]
\end{proof}

\begin{remark}
	For $q=2$ we have $a = M \bmod 2 \in \{0,1\}$, and the bound in Lemma \ref{lem:plot_gen}
	reduces to
	\[
	N(D) \ge
	\frac{4}{M^2 - a(2-a)}
	\sum_{1 \le i < j \le M} [D]_{i,j}.
	\]
	In particular,
	$$
	N(D) \ge
	\begin{cases}
		\displaystyle \frac{4}{M^2}
		\displaystyle \sum_{1 \le i < j \le M} [D]_{i,j}, & \text{if $M$ is even},\\[8pt]
		\displaystyle \frac{4}{M^2-1}
		\displaystyle \sum_{1 \le i < j \le M} [D]_{i,j}, & \text{if $M$ is odd}.
	\end{cases}
	$$
\end{remark}

%%%%%%%%%%%%%%%%%%%%%%%%%%%%%%%%%%%%%%%%%%%%%%%%%%%%%%%%%%%%%%

However, this general bound requires knowledge of all pairwise distances. Therefore, we present an alternative form of the Plotkin bound, expressed in terms of the minimum distances $d_d$ and $d_f$, as defined in Section \ref{main}. %Furthermore, the following bound is for the FCCs over $\mathbb{F}_q$.

\begin{theorem}\label{thm:Plotkin_bound}
    For an $(f\!:d_d,d_f)$-FCC over $\mathbb{F}_q$, where $f:\mathbb{F}_q^k \rightarrow \mathrm{Im}(f)$, we have
    $$r_f(k: d_d,d_f)\geq \frac{1}{q^{k-1}(q-1)}\left((L-1)d_d+(q^k-L)d_f\right)-k,$$
    where $L=\max_{\alpha\in \mathrm{Im}(f)}|f^{-1}(\alpha)|$ and $d_f>d_d$.
\end{theorem}
\begin{proof}
  Let $\mathfrak{C}$ be an $(f\!:d_d,d_f)$-FCC over $\mathbb{F}_q$, consisting of $q^k$ codewords, each of length $n=k + r$, with minimum Hamming distance $d_d$, and minimum Hamming distance $d_f$ between any two codewords corresponding to different function values.
We prove this bound by bounding the total Hamming distance between all pairs of distinct codewords in $\mathfrak{C}$, i.e., $\sum _{(x,y)\in \mathfrak{C}^{2},x\neq y}d(x,y)$.

Consider a codeword $x \in \mathfrak{C}$, with $u_x = \mathfrak{C}^{-1}(x) \in \mathbb{F}_q^k$ as the corresponding message vector, and let the associated function value be $f(u_x) = f_x$.
Then the distance of $x$ from all the codewords $y \in \mathfrak{C}$ such that $f_y=f_x$, is at least $d_d$ and its  distance from all the codewords $z \in \mathfrak{C}$ such that $f_z\neq f_x$, is at least $d_f$. Therefore, 
\begin{align*}
    S(x)&=\sum_{y\in \mathfrak{C}, y\neq x}d(x,y)\\
    &\geq \left(\ |f^{-1}(f_x)|-1 \right) d_d + \left(q^k-|f^{-1}(f_x)|\ \right)d_f \\ 
    &=q^kd_f-d_d-(d_f-d_d)\ |f^{-1}(f_x)|.
\end{align*}
%where $f^{-1}(\alpha)=\{u\in \mathbb{F}_q^k \mid f(u)=\alpha\}$. 

Since $d_f>d_d$ and $L=\max_{\alpha\in \mathrm{Im}(f)} |f^{-1}(\alpha)|$, for all $x\in \mathfrak{C}$, we have 
$$S(x)\geq q^kd_f-d_d-(d_f-d_d) L=(L-1)d_d+(q^k-L)d_f.$$
Therefore,
\begin{equation}\label{PB:eq1}
\sum _{(x,y) \in \mathfrak{C}^{2},x\neq y}d(x,y) =\sum_{x\in \mathfrak{C}} S(x) \ge q^k \left((L-1)d_d+(q^k-L)d_f \right).   
\end{equation}

To upper bound this sum, we check the contribution of each coordinate. For a fixed coordinate position, let $n_a$ denote the number of codewords in which that coordinate equals the symbol $a \in \mathbb{F}_q$. The number of pairs that differ in that coordinate is
$$
\sum_{a \in \mathbb{F}_q} n_a (q^k - n_a) = q^{2k} - \sum_{a \in \mathbb{F}_q} n_a^2.
$$
This is maximized when the codewords are evenly distributed, i.e., each symbol appears equally often, i.e., when $n_a = q^k/q=q^{k-1}$ for all $a$, i.e., when,
$$
\sum_{a \in \mathbb{F}_q} n_a^2 = \frac{q^{2k}}{q}=q^{2k-1}.
$$
Thus, the number of differing pairs per coordinate is at most
$
q^{2k} - q^{2k-1} = q^{2k-1} (q -1).
$
Summing over all $n$ coordinates, the total pairwise distance is at most
\begin{equation}\label{PB:eq2}
    \sum _{(x,y)\in \mathfrak{C}^{2},x\neq y}d(x,y)\leq
n q^{2k-1} (q -1).
\end{equation}
Now from \eqref{PB:eq1} and \eqref{PB:eq2}, we obtain
$
q^k \left((L-1)d_d+(q^k-L)d_f \right) \leq
n q^{2k-1} (q -1).
$
Since, $n=k+r$, we get 

 $$r\geq \frac{1}{q^{k-1}(q-1)}\left((L-1)d_d+(q^k-L)d_f\right)-k.$$
Hence the bound.
\end{proof}

The following example demonstrates that the bounds in Theorems~\ref{thm1.2} and~\ref{thm:Plotkin_bound} are each tighter than the other for different choices of parameters. 

\begin{example}\label{Ex:PB}
   Consider the function $f:\mathbb{F}_2^4 \rightarrow\{0,1,2,3\}$ defined as $f(u)=wt(u) \bmod 4$, for all $u\in \mathbb{F}_2^4$. For this function we have, $L=\max_{\alpha\in \{0,1,2,3\}} |f^{-1}(\alpha)|=6$.
   We have 
   \begin{center}
\begin{tabular}{ |c|c|c| } 
 \hline
  & Bound in Thm. \ref{thm1.2} & Bound in Thm. \ref{thm:Plotkin_bound} \\ 
  \hline
 $d_d=3, d_f=5$ & $5$ & $4.1\simeq5$ \\ 
$d_d=5, d_f=7$ & $8$ & $7.87\simeq8$ \\
$d_d=7, d_f=9$ & $11$ & $11.6\simeq 12$ \\
$d_d=9, d_f=11$ & $14$ & $15.3\simeq16$ \\ 
 \hline
\end{tabular}
\end{center}
Clearly, the variant of the Plotkin bound from Theorem \ref{thm:Plotkin_bound} is tighter for this function.

Now consider the OR function $g:\mathbb{F}_2^4\rightarrow\{0,1\}$ defined as $g(0000)=0$ and $g(u)=1$ for all $u\in\mathbb{F}_2^4, u \neq0000$. The value of $L$ is $15$.
\begin{center}
\begin{tabular}{ |c|c|c| } 
 \hline
  & Bound in Thm. \ref{thm1.2} & Bound in Thm. \ref{thm:Plotkin_bound} \\ 
  \hline
 $d_d=3, d_f=5$ & $5$ & $1.88\simeq2$ \\ 
$d_d=5, d_f=7$ & $8$ & $5.63\simeq6$ \\
$d_d=7, d_f=9$ & $11$ & $9.275\simeq 10$ \\
$d_d=9, d_f=11$ & $14$ & $13.125\simeq14$ \\ 
 \hline
\end{tabular}
\end{center}
\end{example}

In the variant of the Plotkin bound given in Theorem \ref{thm:Plotkin_bound}, the value of $L$ can be easily obtained for some known classes of functions, such as
\begin{itemize}
    \item \textbf{Hamming weight function:}
    For the Hamming weight function $f:\mathbb{F}_2^k \rightarrow\{0,1,\ldots,k\}$ defined as $f(u)=wt(u)$, for all $u\in \mathbb{F}_2^k$, the value of $L$ is
$$L=\max_{\alpha\in \{0,1,\ldots,k\}} |f^{-1}(\alpha)|=\max_{\alpha\in \{0,1,\ldots,k\}} {k\choose\alpha}={k\choose\lfloor k/2 \rfloor}.$$
\item \textbf{Linear function:} For a linear function $f:\mathbb{F}_q^k \rightarrow \mathbb{F}_q^{\ell}$, we have
$$L=\max_{\alpha\in \mathbb{F}_q^{\ell}} |f^{-1}(\alpha)|=\frac{q^k}{q^{\ell}}=q^{k-\ell}.$$
\end{itemize}

\subsection{Generalized Hamming bound for $(f,t)-FCCs$}

In this subsection, we generalize the Hamming bound for $(f,t)-FCCs$.

\begin{theorem}[Generalization of Hamming bound for $(f,t)$-FCC]\label{thm:HB}
    Consider a function $f:\mathbb{F}_q^k \rightarrow \mathrm{Im}(f)$ with $\mathrm{Im}(f)=\{f_1,f_2,\ldots,f_E\}$. 
   Let $\ell=\min_{i\in [E]} |f^{-1}(f_i)|$, then there exists an $(f,t)$-FCC with length $n$ if  
    $$E\leq \frac{q^n}{\left| \cup_{j=1}^{\ell} B(v_j,t) \right|},$$
    where $v_1, v_2, \ldots, v_{\ell}$ are any distinct  vectors in $\mathbb{F}_q^n$ for which $\left| \cup_{j=1}^{\ell} B(v_j,t) \right|$ is minimum.
\end{theorem}

\begin{proof}
    If an $(f,t)$-FCC code with encoding $\mathcal{C}$ that can correct up to $t$ errors in function values, then the Hamming balls of radius $t$ around any codeword $\mathcal{C}(u)$ where $f(u) = f_i$ for some $i\in[E]$ must not contain any codeword $\mathcal{C}(v)$ for which $f(v)\neq f_i$. That means all the unions $\cup_{u\in f^{-1}(f_i)} B(u,t)$ must be disjoint for each $i\in [E]$. Otherwise, a received word could be decoded to another codeword with different function value, violating the function-correcting property.

Since the partition of $\mathbb{F}_q^k$ done by the preimage sets of $f$ has $E$ sets in it, the total number of vectors covered by these disjoint unions $\cup_{u\in f^{-1}(f_i)} B(u,t), i \in [E]$ is at least
$$E \min_{i\in[E]}\left|\cup_{u\in f^{-1}(f_i)} B(u,t)\right| \geq E \left|\cup_{j=1}^{\ell} B(v_j,t)\right|,$$ 
 where $\ell=\min_{i\in [E]} |f^{-1}(f_i)|$, and $v_1, v_2, \ldots, v_{\ell}$ be any distinct vectors in $\mathbb{F}_q^n$ for which $\left| \cup_{j=1}^{\ell} B(v_j,t) \right|$ is minimum.
This must fit within the total space $\mathbb{F}_q^n$ of size $q^n$, hence the result.
\end{proof}

The bound given in Theorem \ref{thm:HB} works for any function $f$ with image $\mathrm{Im}(f)=\{f_1, f_2, \ldots, f_E\}$ such that $|f^{-1}(f_i)|\geq \ell$ for all $i\in [E]$. %That means that only the sizes of the blocks in its domain partition are needed, not the whole partition.
 If $f:\mathbb{F}_q^k \rightarrow \mathrm{Im}(f)$ is a linear function, %and domain partition is known completely, 
 then $\ell= |f^{-1}(f_i)|$ for all $i\in [E]$.  
 %$$E\leq \frac{q^n}{\left| \cup_{v\in \text{Ker}(f)} B(v,t) \right|},$$
 %%%%%%%%%%%%%%%%%%%%%%%___________%%%%%%%%%%%%%%%%%%%
 
 \begin{example}\label{ex8}
    Consider a function $f$ from $\mathbb{F}_2^4$ to $\mathbb{F}_2^2$ defined as
$$ f:\qquad x \mapsto \begin{bmatrix}  1 & 1& 1& 0 \\ 0 & 1 & 1 & 0  \end{bmatrix} x, $$
for which we have $$f^{-1}(00) = \{0000, 0001, 0110, 0111\},$$  
$$f^{-1}(11) = \{0010, 0011, 0100, 0101\},$$
$$f^{-1}(10) = \{1000, 1001, 1110, 1111\},$$ $$ f^{-1}(01) = \{1100, 1101, 1010, 1011\}\}.$$
Here $\ell=\min_{i\in [4]} |f^{-1}(f_i)|=4$.
From Theorem \ref{thm:HB}, we have 
$$E\leq \frac{2^n}{\left| 
\cup_{j=1}^{4} B(v_j,t) \right|},$$
where $v_1,v_2, v_3, v_4\in \mathbb{F}_2^k$ distinct vectors for which $\left| 
\cup_{j=1}^{4} B(v_j,t) \right|$ is minimum. Also 
\begin{equation}\label{eq:ex8.1}
    E\leq \frac{2^n}{\left| 
\cup_{j=1}^{4} B(v_j,t) \right|}\leq \frac{2^n}{\left| 
\cup_{j=1}^{3} B(v_j,t) \right|}.
\end{equation}
As in $\mathbb{F}_2^n$, the nearest possible distinct vectors will be in a similar form to $v_1=(0,0,\ldots,0), v_2=(1,0,\ldots,0)$ and $v_3=(0,1,0,\ldots,0)$, which leads to the minimum size of the union of Hamming balls around three vectors.   
For $t=2$, from Theorem \ref{thm:3BF2}, we have 
\begin{align*}
& \left|\cup_{j=1}^{3} B(v_j,t)\right|& \\
 &=3\sum_{i=0}^{t} {n \choose i} - 6 \sum_{i=0}^{t-1} {n-1 \choose i} + {n-2 \choose t-1}  + 4 \sum_{i=0}^{t-2} {n-2 \choose i}\\
 &= 3\left(1+n+\frac{n(n-1)}{2}\right)-6(1+(n-1))+(n-2)+4(1)\\
 &=\frac{3}{2}n^2-\frac{7}{2}+5.
\end{align*}
Now, from \eqref{eq:ex8.1}, we have $$ E\leq  \frac{2^n}{\left| 
\cup_{j=1}^{3} B(v_j,t) \right|}=\frac{2^n}{\frac{3}{2}n^2-\frac{7}{2}+5}.$$
That means, $4(\frac{3}{2}n^2-\frac{7}{2}+5)\leq 2^n$ from which we get $n\ge 9$. As verified with the trial-and-error method, the optimal length for this function with $t=2$ is $10$. 
\end{example}
 
 %%%%%%%%%%%%%%%%%%%%%%%___________%%%%%%%%%%%%%%%%%%
 
\begin{corollary}
   Since $\left| 
\cup_{j=1}^{\ell} B(v_j,t) \right|\geq |B(v,t)|$ for any $v\in \mathbb{F}_q^k$ and $|B(v,t)|=\sum_{i=0}^t {n \choose i} (q-1)^i$, we have 
$$E\leq \frac{q^n}{\sum_{i=0}^t {n \choose i} (q-1)^i}.$$
\end{corollary}

If $f$ is a bijection, then $(f,t)$-FCC turns into an ECC with minimum distance $2t+1$. Then $E=q^k, \ell=1$ and the bound given in Theorem \ref{thm:HB} turns into the Hamming bound as
$$q^k\leq \frac{q^n}{|B(v,t)|}=\frac{q^n}{\sum_{i=0}^{t}{n \choose i} (q-1)^i},$$
where $v$ is a vector in $\mathbb{F}_q^k$.

%%%%%%%%%%%%%%%%%%%%%%%%%%%%
\subsection{ Hamming bound for $(f\!:d_d,d_f)$-FCCs}
\begin{theorem}[Generalization of Hamming bound for $(f\!:d_d,d_f)$-FCC] \label{thm:HBfortdtf}
    Consider a function $f:\mathbb{F}_q^k \rightarrow \mathrm{Im}(f)$ with $\mathrm{Im}(f)=\{f_1,f_2,\ldots,f_E\}$. 
   Let $\ell=\min_{i\in [E]} |f^{-1}(f_i)|$ and $v_1, v_2, \ldots, v_{\ell}$ be any distinct vectors in $\mathbb{F}_q^n$ with the condition that $d(v_i,v_j)\geq d_d$ for all $i,j\in[\ell], i\ne j$, for which $\left|\cup_{j=1}^{\ell} B(v_j,t_f) \right|$ is minimum. Then there exists an $(f\!:d_d,d_f)$-FCC with length $n$ if
    $$E\leq \frac{q^n}{\left| 
\cup_{j=1}^{\ell} B(v_j,t_f) \right|}.$$
\end{theorem}

\begin{proof}
    The proof is the same as the proof of Theorem \ref{thm:HB}, with the difference that now each codeword must have a Hamming distance of at least $d_d$ with any other codeword. Therefore, the vectors $v_1, v_2, \ldots, v_{\ell}$ in $\mathbb{F}_q^n$ are selected in such a way that $d(v_i,v_j)\geq d_d$ for all $i,j\in[\ell], i\ne j$ and $\left|\cup_{j=1}^{\ell} B(v_j,t_f) \right|$ is minimum. 
\end{proof}
%%%%%%%%%%%%%%%%%%%%%%------------------------%%%%%%%%%%%%%%%%%%%%%

\begin{example}
	Consider the function $f:\{0,1\}^3 \to \{0,1,2,3\}$ defined by
	\[
	f(x_1,x_2,x_3)=2x_1+(x_2+x_3)\bmod 2.
	\]
	That is $f(000)=f(011)=0, f(010)=f(001)=1, f(100)=f(111)=2, f(101)=f(110)=3.$
	Thus $E=4$ and $\ell=\min_{i\in[E]}|f^{-1}(i)|=2$.  
	For $d_d=3$ and $d_f=5$ (i.e., $t_d=1$ and $t_f=2$), Theorem~\ref{thm:HBfortdtf} gives
	\begin{equation}\label{ex:eq2}
		4 \le \frac{2^n}{\left|B(v_1,2)\cup B(v_2,2)\right|},
	\end{equation}
	where $v_1$ and $v_2$ are two distinct vectors in $\mathbb{F}_2^n$ with $d(v_1,v_2)\ge 3$ and for which 
	$\left|B(v_1,2)\cup B(v_2,2)\right|$ is minimized.
	Using Theorem~\ref{thm:2BF2D3}, we obtain
	\begin{align*}
		\left|B(v_1,2)\cup B(v_2,2)\right|
		&=2\sum_{i=0}^{t_f}\binom{n}{i}
		- 8\sum_{i=0}^{t_f-3}\binom{n-3}{i}
		- 6\binom{n-3}{t_f-2} \\
		&= 2\sum_{i=0}^{2}\binom{n}{i}
		- 8\cdot 0
		- 6\binom{n-3}{0} \\
		&=2\bigl(1+n+\tfrac{n(n-1)}{2}\bigr)-6 \\
		&= n^2+n-4.
	\end{align*}
	Hence, from \eqref{ex:eq2} we get
	$
	n^2+n-4 \le 2^{\,n-2}.
	$
	The smallest integer satisfying this inequality is $n=9$.  
	Therefore, any $(f\!:3,5)$-FCC must have length $n\ge 9$.
	An explicit $(f\!:3,5)$-FCC of length $9$ obtained by our construction method is given below.
	\[
	\begin{array}{ccccc}
		u\in\mathbb{F}_2^3 & f(u) & \text{ECC-parity} &
		\text{FCC-parity} & \text{codeword} \\ 
		000 & 0 & 000 & 000 & 000000000 \\
		011 & 0 & 110 & 000 & 011110000 \\
		010 & 1 & 101 & 110 & 010101110 \\
		001 & 1 & 011 & 110 & 001011110 \\
		100 & 2 & 110 & 101 & 100110101 \\
		111 & 2 & 000 & 101 & 111000101 \\
		101 & 3 & 101 & 011 & 101101011 \\
		110 & 3 & 011 & 011 & 110011011
	\end{array}
	\]
	This meets the lower bound and is therefore optimal.
\end{example}

%%%%%%%%%%%%%%%%%%%%%%------------------------%%%%%%%%%%%%%%%%%%%%%
\begin{corollary}
     Consider a function $f:\mathbb{F}_q^k \rightarrow \mathrm{Im}(f)$ with $\mathrm{Im}(f)=\{f_1,f_2,\ldots,f_E\}$, and $\ell=\min_{i\in [E]} |f^{-1}(f_i)|$. Then there exists an $(f\!:d_d,d_f)$-FCC with length $n$ if 
    $$E\leq \frac{q^n}{\ell\left( 
\sum_{i=0}^{t_d} {n \choose i} (q-1)^i\right)}.$$
\end{corollary}

\begin{proof}
    In Theorem \ref{thm:HBfortdtf}, the vectors $v_1, v_2, \ldots, v_{\ell}$ are chosen in such a way that $d(v_i,v_j)\geq d_d$ for all $i,j\in[\ell], i\ne j$. Since $t_f>t_d$, the union $\cup_{j=1}^{\ell} B(v_j, t_f)$ contains all the balls $B(v_j, t_d), j \in[\ell]$, which are pairwise disjoint. Therefore, we have 
    \begin{align*}
        \left|\cup_{j=1}^{\ell} B(v_j, t_f)\right| &\geq \left|\cup_{j=1}^{\ell} B(v_j, t_d)\right|=\ell \left| B(v_1, t_d)\right| \\
        &= \ell\left( 
\sum_{i=0}^{t_d} {n \choose i} (q-1)^i\right).
    \end{align*}
Hence, the result.
\end{proof}

\section{Conclusion}
\label{conclusion}
This work presents a generalized framework for function-correcting codes, focusing on scenarios where the function value requires stronger protection than the data itself. Our construction method enables the design of such codes, illustrated through specific examples. By leveraging connections to irregular distance codes, we derive bounds on redundancy and further explore the structure of linear function-correcting codes, with particular attention to linear functions. Finally, we generalize the Plotkin and Hamming bounds from classical error-correcting codes to the setting of function-correcting codes.  Future work may include finding better lower bounds that also take the size of the message space into consideration.

\section*{Acknowledgment}
This work was supported by a joint project grant to Aalto University and Chalmers University of Technology (PIs A. Graell i Amat and C. Hollanti) from the Wallenberg AI, Autonomous Systems and Software Program, and additionally by the Science and Engineering Research Board (SERB) of the Department of Science and Technology (DST), Government of India, through the J. C. Bose National Fellowship to Prof. B. Sundar Rajan.

\appendices
\section{Counting the number of vectors in the union of balls of radius $t$}
\begin{theorem}\label{thm:2B}
    Consider two Hamming balls $B(u,t)$ and $B(v,t)$ for $u,v\in \mathbb{F}_q^n$ such that $d(u,v)=1$. Then 
    $$|B(u,t) \cup B(v,t)| = 2\sum_{i=0}^{t} {n \choose i} (q-1)^i-q\sum_{i=0}^{t-1} {n-1 \choose i} (q-1)^i.$$
\end{theorem}
\begin{proof}
    Without loss of generality, assume that $u=(0,0,\ldots,0)$ and $v=(1,0,\ldots,0)$ in $\mathbb{F}_q^n$. Since we have 
    $$|B(u,t) \cup B(v,t)|=|B(u,t)|+|B(v,t)|-|B(u,t) \cap B(v,t)|,$$
    first we will compute $|B(u,t) \cap B(v,t)|$. A vector $x=(x_1, x_2, \ldots, x_n)$ lies in $B(u,t) \cap B(v,t)$ if $d(x,u) \leq t$ and $d(x,v)\leq t$. For such a vector $x$, denoting $\text{wt}((x_2, \ldots,x_n))$ as $w$, we have the following cases.
    \begin{itemize}
        \item \textbf{Case 1:} If $x_1=0$, then $d(x,u)=\text{wt}((x_2, \ldots,x_n))=w$ and $d(x,v)=1+w$. Since $d(x,u) \leq t$ and $d(x,v)\leq t$, we have $w\leq t-1$. The choices for the vector $x$ are 
        $$\#_1=\sum_{i=0}^{t-1} {n-1 \choose i} (q-1)^i.$$
         \item \textbf{Case 2:} If $x_1=1$, then $d(x,u)=1+w$ and $d(x,v)=w$. Since $d(x,u) \leq t$ and $d(x,v)\leq t$, we have $w\leq t-1$. Again, the choices for the vector $x$ are 
        $$\#_2=\sum_{i=0}^{t-1} {n-1 \choose i} (q-1)^i.$$
         \item \textbf{Case 3:} If $x_1\ne 0$ and $x_1\ne 1$, then $d(x,u)=1+w$ and $d(x,v)=1+w$. Since $d(x,u) \leq t$ and $d(x,v)\leq t$, we have $w\leq t-1$. Again the choices for the vector $x$ are 
        $$\#_3=\sum_{i=0}^{t-1} {n-1 \choose i} (q-1)^i,$$
        for each $x_1\in \mathbb{F}_q \setminus \{0,1\}$.
    \end{itemize}
    Therefore, we have 
    $$|B(u,t)\cap B(v,t)|=\#_1+\#_2+(q-2)\#_3%=(1+1+q-2)\left(\sum_{i=0}^{t-1} {n-1 \choose i} (q-1)^i\right)
    = q\sum_{i=0}^{t-1} {n-1 \choose i} (q-1)^i,$$
    and 
    $$|B(u,t) \cup B(v,t)|=2\sum_{i=0}^{t} {n \choose i} (q-1)^i-q\sum_{i=0}^{t-1} {n-1 \choose i} (q-1)^i$$
    
\end{proof}

\begin{lemma}\label{lem:2BF2cap}
    Consider two vectors $u_1,u_2\in \mathbb{F}_2^n$ such that $d(u_1,u_2)=2$. Then 
    $$|B(u_1,t) \cap B(u_2,t)| =2\sum_{i=0}^{t-1} {n-1 \choose i}.$$
\end{lemma}
\begin{proof}
Without loss of generality, consider $u_1=(0,0,\ldots, 0)$ and $u_2=(1,1,0,\ldots, 0)$. A vector $x=(x_1, x_2, \ldots, x_n)$ lies in the intersection if $d(x,u_1) \leq t$ and $d(x,u_2)\leq t$. For such a vector $x$, denoting $\text{wt}((x_3, \ldots,x_n))$ as $w$, we have the following cases.
    \begin{itemize}
        \item \textbf{Case 1:} If $x_1=0$ and $x_2=0$, then $d(x,u_1)=w$ and $d(x,u_2)=2+w$, and we have $w\leq t-2$. The choices for the vector $x$ are 
        $$\#_1=\sum_{i=0}^{t-2} {n-2 \choose i}.$$
        \item \textbf{Case 2:} If $x_1=1$ and $x_2=0$, then $d(x,u_1)=1+w$ and $d(x,u_2)=1+w$, and we have $w\leq t-1$. The choices for the vector $x$ are 
        $$\#_2=\sum_{i=0}^{t-1} {n-2 \choose i}.$$
        \item \textbf{Case 3:} If $x_1=0$ and $x_2=1$, then as similar to Case 2, the choices for the vector $x$ are 
        $\#_3=\#_2=\sum_{i=0}^{t-1} {n-2 \choose i}.$
        \item \textbf{Case 4:} If $x_1=1$ and $x_2=1$, then $d(x,u_1)=2+w$ and $d(x,u_2)=w$, and we have $w\leq t-2$. The choices for the vector $x$ are 
        $\#_4=\#_1=\sum_{i=0}^{t-2} {n-2 \choose i}.$
    \end{itemize} 
    Therefore, we have 
    \begin{align*}
        |B(u_1,t)\cap B(u_2,t)|&=2\#_1+2\#_2
    = 2\left(\sum_{i=0}^{t-1} {n-2 \choose i} + \sum_{i=0}^{t-2} {n-2 \choose i}\right)\\
    &=2\left(\sum_{i=0}^{t-1} {n-2 \choose i} + \sum_{i=0}^{t-1} {n-2 \choose i-1}\right)=2\sum_{i=0}^{t-1} {n-1 \choose i}.
    \end{align*}
\end{proof}

\begin{theorem}\label{thm:3BF2}
    Consider three distinct vectors $u_1,u_2, u_3\in \mathbb{F}_2^n$ with pairwise distances $1,1,2$. Then 
    $$
    |B(u_1,t) \cup B(u_2,t) \cup B(u_3,t)| = 3\sum_{i=0}^{t} {n \choose i} - 6 \sum_{i=0}^{t-1} {n-1 \choose i}  + {n-2 \choose t-1}  +4 \sum_{i=0}^{t-2} {n-2 \choose i}.
   $$
\end{theorem}
 \begin{proof}
     Let the vectors be $u_1=(0,0,\ldots, 0), u_2=(1,0,\ldots, 0)$ and $u_3=(0,1,0,\ldots, 0)$, without loss of generality. First, we find $|B(u_1,t)\cap B(u_2,t) \cap B(u_3,t)|$. A vector $x=(x_1, x_2, \ldots, x_n)$ lies in the intersection if $d(x,u_i) \leq t$ for $i\in[3]$. For such a vector $x$, denoting $\text{wt}((x_3, \ldots,x_n))$ as $w$, we have the following cases.
    \begin{itemize}
        \item \textbf{Case 1:} If $x_1=0$ and $x_2=0$, then $d(x,u_1)=w, d(x,u_2)=1+w$ and $d(x,u_3)=1+w$. Since $d(x,u_i) \leq t$ for $i\in[3]$, we have $w\leq t-1$. The choices for the vector $x$ are 
        $$\#_1=\sum_{i=0}^{t-1} {n-2 \choose i}.$$
        \item \textbf{Case 2:} If $x_1=1$ and $x_2=0$, then $d(x,u_1)=1+w, d(x,u_2)=w$ and $d(x,u_3)=2+w$. Since $d(x,u_i) \leq t$ for $i\in[3]$, we have $w\leq t-2$. The choices for the vector $x$ are 
        $$\#_2=\sum_{i=0}^{t-2} {n-2 \choose i}.$$
        \item \textbf{Case 3:} If $x_1=0$ and $x_2=1$, then as similar to Case 2, the choices for the vector $x$ are 
        $\#_3=\#_2=\sum_{i=0}^{t-2} {n-2 \choose i}.$
        \item \textbf{Case 4:} If $x_1=1$ and $x_2=1$, then $d(x,u_1)=2+w, d(x,u_2)=1+w$ and $d(x,u_3)=1+w$. Since $d(x,u_i) \leq t$ for $i\in[3]$, we have $w\leq t-2$. The choices for the vector $x$ are 
        $\#_4=\#_2=\sum_{i=0}^{t-2} {n-2 \choose i}.$
    \end{itemize} 
    Therefore, we have 
   $$
        |B(u_1,t)\cap B(u_2,t) \cap B(u_3,t)|=\#_1+3\#_2
    = \sum_{i=0}^{t-1} {n-2 \choose i} +3 \sum_{i=0}^{t-2} {n-2 \choose i}  ={n-2 \choose t-1} + 4 \sum_{i=0}^{t-2} {n-2 \choose i},
    $$
    and 
    \begin{align*}
      &|B(u_1,t)\cup B(u_2,t) \cup B(u_3,t)| \\
      &
    = |B(u_1,t)| + |B(u_2,t)| + |B(u_3,t)|-|B(u_1,t) \cap B(u_2,t)| - |B(u_2,t) \cap B(u_3,t)| - |B(u_1,t) \cap B(u_3,t)| \\
    & \quad +|B(u_1,t) \cap B(u_2,t) \cap B(u_3,t)| \\
    &=3\sum_{i=0}^{t} {n \choose i} - 3 \left( 2 \sum_{i=0}^{t-1} {n-1 \choose i}\right)+{n-2 \choose t-1} + 4 \sum_{i=0}^{t-2} {n-2 \choose i}\\
    &=3\sum_{i=0}^{t} {n \choose i} - 6 \sum_{i=0}^{t-1} {n-1 \choose i}+{n-2 \choose t-1} + 4 \sum_{i=0}^{t-2} {n-2 \choose i}.
    \end{align*}
 \end{proof}

\begin{theorem}\label{thm:2BF2D3}
    Consider two vectors $u_1,u_2\in \mathbb{F}_2^n$ such that $d(u_1,u_2)=3$. Then for $t\geq 2$,
    $$|B(u_1,t) \cup B(u_2,t)| =2\sum_{i=0}^{t}{n \choose i} - 8\sum_{i=0}^{t-3} {n-3 \choose i} - 6{n-3 \choose t-2}.$$
\end{theorem}
\begin{proof}
Without loss of generality, consider $u_1=(0,0,\ldots, 0)$ and $u_2=(1,1,1,0,\ldots, 0)$. A vector $x=(x_1, x_2, \ldots, x_n)$ lies in the intersection if $d(x,u_1) \leq t$ and $d(x,u_2)\leq t$. For such a vector $x$, denoting $\text{wt}((x_4, \ldots,x_n))$ as $w$, following the same method as in Lemma \ref{lem:2BF2cap}, we have the following cases.

\[
\begin{array}{ccccc}
x_1,x_2,x_3 & d(x,u_1) & d(x, u_2)&
\text{Condition} & \text{Choices for} \ x \\ \hline
000 & w & 3+w & w\le t-3 & \#_2 \\
100 & 1+w & 2+w & w \le t-2 & \#_1 \\
010 & 1+w & 2+w & w \le t-2 & \#_1 \\
001 & 1+w & 2+w & w \le t-2 & \#_1 \\
110 & 2+w & 1+w & w \le t-2 & \#_1 \\
101 & 2+w & 1+w & w \le t-2 & \#_1 \\
011 & 2+w & 1+w & w \le t-2 & \#_1 \\
111 & 3+w & w & w \le t-3 & \#_2 \end{array}
\]
where $\#_1=\sum_{i=0}^{t-2} {n-3 \choose i}$ and $\#_2=\sum_{i=0}^{t-3} {n-3 \choose i}$.
    Therefore, we have 
    \begin{align*}
        |B(u_1,t)\cap B(u_2,t)|&=6\#_1+2\#_2
    = 6\sum_{i=0}^{t-2} {n-3 \choose i} + 2 \sum_{i=0}^{t-3} {n-3 \choose i}=8\sum_{i=0}^{t-3} {n-3 \choose i} + 6 {n-3 \choose t-2},
    \end{align*}
    and 
    \begin{align*}
        |B(u_1,t)\cup B(u_2,t)|&= |B(u_1,t)| + |B(u_2, t)| - |B(u_1,t) \cap B(u_2,t)| \\
    &=2\sum_{i=0}^t {n\choose i}-8\sum_{i=0}^{t-3} {n-3 \choose i} - 6 {n-3 \choose t-2}.
    \end{align*}
\end{proof}


%\section*{Acknowledgment}


\ifCLASSOPTIONcaptionsoff
  \newpage
\fi



%\bibliographystyle{IEEEtran}

\end{document}